\documentclass[10pt,conference,final,twocolumn]{IEEEtran}

\makeatletter
\newif\if@restonecol
\makeatother

\usepackage{eucal}
\usepackage{graphicx}
\usepackage{times}
\usepackage{latexsym}
\usepackage{amsmath}
\usepackage{amssymb}
\usepackage{amsthm}
\usepackage{epsfig}
\usepackage{cite}
\usepackage{cases}
\usepackage{amsfonts}
\usepackage{indentfirst}
\usepackage{epsfig}
\usepackage{amsopn}
\usepackage{bm}
\usepackage[linesnumbered,ruled]{algorithm2e}
\usepackage{subfig}
\usepackage{eufrak}
\usepackage{stfloats}
\usepackage{booktabs}
\usepackage{slashbox}
\usepackage{textcomp}

\newtheorem{lemma}{Lemma}
\newtheorem{theorem}{Theorem}

\newtheorem{comment}{Comment}[section]
\newcommand{\Tr}{\mathsf{Tr}}
\newcommand{\Real}{\mathsf{Re}}

\begin{document}

\title{MIMO Beamforming Design towards Maximizing Mutual Information in Wireless Sensor Network}
\author{Yang Liu, \   Jing Li, \  and \ Xuanxuan Lu \\
Electrical and Computer Engineering Department, Lehigh University, Bethlehem, PA 18015, USA \\
Email: \{yal210@lehigh.edu, jingli@ece.lehigh.edu, xul311@lehigh.edu\}
}


\maketitle

\footnotetext{Part of this work appeared in our conference paper \cite{bib:YangLiu_ICC}, which discussed the cyclic multi-block approach in the special case of scalar source.}

\begin{abstract}

This paper considers joint beamformer design towards maximizing the mutual information in a coherent wireless sensor network with noisy observation and multiple antennae. Leveraging the weighted minimum mean square error and block coordinate ascent (BCA) framework, we propose two new and efficient methods: batch-mode BCA and  cyclic multi-block BCA. The existing batch-mode approaches require stringent conditions such as diagonal channel matrices and positive definite second-order matrices, and are therefore inapplicable to our problem. Our match-mode BCA overcomes the previous limitations via a general second-order cone programming  formation, and exhibits a strong convergence property which we have rigorously proven. The existing multi-block approaches rely on numerical solvers to handle the subproblems and some render good performance only at high signal-to-noise ratios. Exploiting the convexity of the trust-region subproblem for the convex case, our multi-block BCA significantly reduces the complexity and enhances the previous results by providing an analytical expression for the energy-preserving optimal solution. Analysis and simulations confirm the advantages of the proposed methods. 
\end{abstract}


\section{Introduction}
\label{sec:introduction}

Consider a multi-input multi-output (MIMO) wireless sensor network deployed for monitoring or surveillance purpose. Depending on the individual capacity, each sensor may be equipped with a different number of antennae and provisioned with a different power constraint. When an event $\mathbf{s}$ occurs, it is highly likely that multiple sensors may have sensed it, but each observation will likely be distorted (with independent distortion resulting from, for example, thermal noise and circuitry noise). A central question lying in this line of WSN scenarios is how the sensors should communicate with the fusion center (FC) in reliable and efficient manner.

From the information theory perspective, this type of problems may be regarded as a form of the so-called CEO (chief executive officer) problem \cite{bib:CEO_prob}, where a group of agents feed individually biased observations to a CEO, who tries to recover the original source as much as possible. From the communication perspective, what's of particular interest is the joint transceiver design between the multiple sensors and the fusion center to achieve the maximum communication efficiency.

A good number of studies exist in the literature on this and similar system models. Most of them target minimizing the mean square error (MSE, which measures the amount of signal distortion), an important criterion widely adopted in signal estimation (e.g. \cite{bib:Sensor_compress_1, bib:Sensor_compress_2, bib:Sensor_compress_3, bib:FangLi_2, bib:sensor_network_Cui, bib:sensor_network_Hamid, bib:J1_Yang_to_be_submitted, bib:JunFang_MI, bib:YangLiu_ICC}), and much fewer considered mutual information (MI,  which measures the amount of uncertainty of the source released at the receiver). Mutual information not only is important in its own right, but also presents useful relevance to MSE and signal-to-noise ratio (SNR) \cite{bib:WMMSE_original} \cite{bib:WMMSE_QingjiangShi}.

The goal of this paper is to develop new and efficient ways to design MIMO transceivers (beamformers) that enable maximal end-to-end mutual information in a noisy-observation multi-sensor wireless network. Since beamformer design problems are usually nonconvex and hence difficult, and since a large set of variables and constraints are typically involved, it has become a standard approach to leverage the block coordinate descent/ascent (BCD/BCA) tool to break down the original problem into a set of smaller subproblems (blocks) involving fewer variables each. Nonetheless,  each subproblem can still be difficult to handle, and the specific method to address them  can make a huge difference in performance and especially in complexity (and speed). Further, a small difference in the system model (e.g. separate power constraints vs a total power constraint) can sometimes lead to a rather different optimization problem, causing a previous method to become rather inefficient or totally inapplicable.

Several early beamformer papers for MIMO WSN have focused on compression or signal space reduction (i.e. precoders with rate $<1$) (e.g. \cite{bib:Sensor_compress_1, bib:Sensor_compress_2, bib:Sensor_compress_3, bib:FangLi_2}). 
Recent studies consider more general transceiver design where the beamformer can have an arbitrary (positive) rate, and have investigated a variety of system setups. Specifically, a very general system model for the WSN noisy-sensor observation problem was introduced in \cite{bib:sensor_network_Cui}, whose parameters include channel fading and additive noise, separate power constraints, and either orthogonal or coherent multiple access channels (MAC). Due to the difficulty in attacking the general model, \cite{bib:sensor_network_Cui} instead provided solutions to specific (but nonetheless very interesting) cases, including the case when the source is a scalar, when the sensor-FC channels are noiseless, and when the sensor-FC channels are non-interferenced (i.e. sensor-FC channel fading matrices are diagonal matrices). Following this general model, a highly noteworthy work is \cite{bib:sensor_network_Hamid},  which developed a very useful BCD-based method targeting MMSE beamformers for  coherent MAC. The method developed in \cite{bib:sensor_network_Hamid} is relatively general and capable of obtaining numerical beamformers for most of the cases --- except when a second-order coefficient matrix for the quadratic optimization problem becomes singular (in which case the method with zero Lagrangian multiplier will fail). A more recent paper \cite{bib:J1_Yang_to_be_submitted}  developed a different and even more general iterative BCD-based  beamformer design method targeting MMSE for the same system model. In addition to solving all the cases，\cite{bib:J1_Yang_to_be_submitted} also demonstrated the convergence of any essentially-cyclic block-updating rule (including the specific block-updating rule used in \cite{bib:sensor_network_Hamid}). 

These previous excellent studies all take MSE as the optimization objective. 
Since mutual information presents an essential indicator for communication efficiency, it is also highly interesting to see how beamformers can be designed to achieve maximal MI. 
The topic is somewhat more involved than the target MSE, and also appears less researched. The only paper we could find in the WSN literature targeting maximal MI beamformer design is a recent paper \cite{bib:JunFang_MI}, which considered an orthogonal sensor-FC channel (a model drastically different from the coherent MAC model considered here). 

It is worth noting that there is a coherent MIMO amplify-forward (AF) model developed in the recent relay literature \cite{bib:MIMO_AF_Relay, bib:CWXing, bib:TSTINR}. The model thereof consists of several source-destination pairs communicating via the help of a group of common relays performing amplify-and-forward (AF). The source, relay and destination nodes in the relay system are like the sources/events, sensors and fusion centers in our sensing system; the difference is that in the sensing system, the first hop is passive transmission with no precoder performed at the transmitter. Since the relay model subsumes our sensing model (by setting the number of source-destine pairs to one), several interesting ideas developed thereof can be relevant. However, as will be shown later (in Subsection \ref{subsec:contribution}, Comment \ref{cmt:3_BCA}), Comment \ref{cmt:convergence}, Table \ref{tab:runningtime1}, Figure \ref{fig:comparisonTSTINR}), our proposed methods are either more efficient or more effective than all of these existing approaches.

\subsection{Contribution of This Paper}
\label{subsec:contribution}

The key contribution of this paper is the development of two new beamforming design methods targeting maximal MI for the noisy-observation multi-sensor system, and the proof of strong convergence results.  Unlike many previous approaches that are limited to only certain models \cite{bib:Sensor_compress_1, bib:Sensor_compress_2, bib:Sensor_compress_3, bib:FangLi_2, bib:sensor_network_Cui, bib:YangLiu_ICC} (where signal dimension has specific constraints, or channel has no fading or noise), the methods developed here are general, and applicable to beamformers with any dimensions and channels with noise and fading. For the rather general sensing model considered in Fig. \ref{fig:sysmodel}, by leveraging the tool of weighted MMSE (WMMSE)\cite{bib:WMMSE_original} \cite{bib:WMMSE_QingjiangShi}, we decompose the original problem into three subproblems that optimize the weight matrix, the postcoder and the set of precoders, respectively. We present analytical solutions for the first two subproblems, and for the third one where closed-form analysis is intractable, we develop two new and efficient methods to solve them.

(i) The first method is a batch-mode approach that designs all the precoders at once. Specifically, we succeeded in transforming the subproblem into a rather general second-order cone programming (SOCP) formulation, which then lends itself to relatively efficient numerical solvers. Not many batch-mode approaches have been previously attempted, and the solutions in the literature are either applicable to only simpler special channel cases \cite{bib:sensor_network_Cui}\footnote{The method in \cite{bib:sensor_network_Cui} is only applicable to the special case where all the channel matrices are diagonal (i.e. parallel channels). Our SOCP approach is applicable to all channel matrices, and runs faster than the semidefinite programming method in \cite{bib:sensor_network_Cui}.} or require additional assumptions \cite{bib:CWXing}\footnote{\cite{bib:CWXing} requires that the second-order matrices in both the objective and the constraint functions be positive definite, which is actually a fairly stringent requirement that cannot be met in our system model. In comparison, our batch-mode approach only requires positive semidefinite second-order matrices which is always satisfied.}. The new batch-mode approach has overcome the previous limitations by providing solution that has a broader usability and runs faster. It also allows us to perform rigorous convergence analysis which has been attempted but nevertheless missing in the previous studies (i.e. only conjectures/discussions but not proofs were presented in the literature, e.g. \cite{bib:MIMO_AF_Relay} \cite{bib:CWXing}). Our convergence analysis confirms that the limit points of our solution satisfy the Karush-Kuhn-Tucker (KKT) condition of the original MI maximization problem.

(ii) The second method is a cyclic multi-block approach that further decomposes the task of optimizing the beamformers at all the sensors into multiple subproblems. 
The subproblems are shown to be convex trust-region subproblems. 
The previous approaches inevitably rely on numerical solvers, such as semidefinite relaxation (SDR) \cite{bib:MIMO_AF_Relay}, 
or the conventional numerical solution to general trust-region surbproblem \cite{bib:TSTINR}, to solve these subproblems, and some of them are only applicable to strictly convex subproblems \cite{bib:sensor_network_Hamid}. Exploiting the convexity, we successfully enhance the standard solution to the trust-region subproblem for the convex case by developing near closed-form solutions (possibly up to a simple one-dimension bisection search). Compared to the conventional numerical methods to solve trust-region subproblems in the literature, our new analytical result has a significantly lower complexity, and provides a better insight into the solution. It clearly delineates the cases for unique optimal solutions and multiple optimal solutions, and the in the latter case, can provide the solution with the least energy consumption (with minimum $l_2$-norm). For scalar sources, a completely closed-form solution is also derived, which eliminates the eigonvalue decomposition or bisection search. These results are not known previously in the literature and greatly improves the efficiency of beamformer design.

\emph{Notations}: We use bold lowercases and bold uppercases to denote complex vectors and  complex matrices, respectively. $\mathbf{0}$, $\mathbf{O}_{m\times n}$, and $\mathbf{I}_m$ represent the all-zero vector, the all-zero matrix of dimension $m\times n$, and the identity matrix of order $m$, respectively. $\mathbf{A}^T$, $\mathbf{A}^{\ast}$, $\mathbf{A}^H$, and $\mathbf{A}^{\dagger}$ represent the transpose, the conjugate, the conjugate transpose (Hermitian transpose), and the Moore-Penrose pseudoinverse, respectively, of a complex matrix $\mathbf{A}$. $\Tr\{\cdot\}$ represents the trace operation of a square matrix. $|\cdot|$ represents the modulus of a complex scalar, and $\|\cdot\|_2$ represents the $l_2$-norm of a complex vector. $vec(\cdot)$ represents the vectorization operation of a matrix, which is performed by packing the columns of a matrix into a single long column. $\otimes$ denotes the Kronecker product. $\mathsf{Diag}\{\mathbf{A}_1,\cdots,\mathbf{A}_n\}$ denotes the block diagonal matrix with its $i$-th diagonal block being the square complex matrix $\mathbf{A}_i$, $i\in\{1,\cdots,n\}$. $\mathcal{H}^n_{+}$ and $\mathcal{H}^n_{++}$ represent the cones of positive semidefinite and positive definite matrices of dimension $n$, respectively, and $\succeq0$ and $\succ0$ mean that a square complex matrix belongs to $\mathcal{H}^n_{+}$ and $\mathcal{H}^n_{++}$, respectively. $\mathcal{R}(\mathbf{Q})$ and $\mathcal{N}(\mathbf{A})$ denote the range space and null space of a matrix $\mathbf{A}$ respectively.  $\mathsf{Re}\{x\}$ means taking the real part of a complex value $x$.

\begin{figure}[htb]
\centerline{
\includegraphics[width=0.5\textwidth]{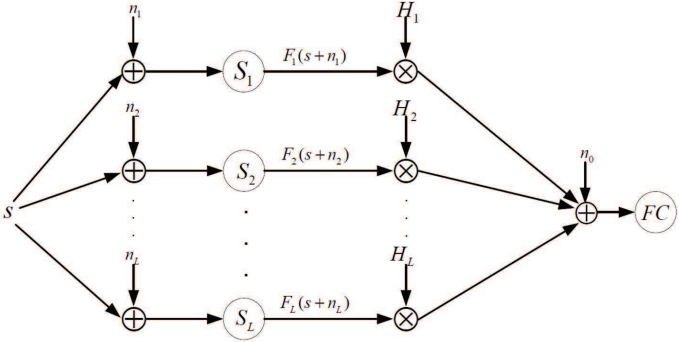}
}
\caption{Multi-Sensor System Model}
\label{fig:sysmodel}
\end{figure}


\section{System Model}
\label{sec:system model}

As illustrated in Fig.\ref{fig:sysmodel}, the MIMO wireless sensing system we consider include $L$ sensors, each equipped with $N_i$ antennae ($i\!=\!1,2,\cdots,L$),  and a fusion center equipped with $M$ antennae. Following the convention, we assume the source signal $\mathbf{s}$ to be a complex Gaussian\footnote{Gaussian source is always assumed in this line of research, as it is the only source type that enables a closed-form MI expression on AWGN channels. Gaussian source is also the most efficient, as it  achieves the capacity of AWGN channels, and provides an upper bound on the amount of information perceived at the fusion center.} vector of dimension $K$, i.e. $\mathbf{s}\sim\mathcal{CN}(\mathbf{0},\mathbf{\Sigma}_{\mathbf{s}})$ where $\mathbf{\Sigma}_{\mathbf{s}}$ is the positive definite covariance matrix. 
Due to the interference from the surroundings and/or the thermal noise from the sensing device, the observations made at each sensor tend to be individually contaminated. We model the sensing distortion using the complex Gaussian distribution: $\mathbf{n}_i\sim\mathcal{CN}(\mathbf{0},  
\mathbf{\Sigma}_i)$, where $\mathbf{\Sigma}_i\in\mathbb{C}^{K\times K}$ is the covariance matrix and   
$\mathbf{n}_i$'s are mutually uncorrelated for $i\!=\!1,2,\cdots,L$. Hence, the $i$-th sensor observes $\mathbf{s}\!+\!\mathbf{n}_i$.

Each sensor applies a linear beamformer $\mathbf{F}_i$ to its observation before putting it through the coherence MAC to the FC. Let $\{\mathbf{H}_i\}\in\mathbb{C}^{M\times N_i}$  be the channel fading matrix between the $i$-th sensor and the FC. $\{\mathbf{H}_i\}$ is assumed to be known to the FC (as it can be obtained via advanced channel estimation techniques). Let $\mathbf{n}_0$ be the receiver-end complex additive white Gaussian noise (AWGN), $\mathbf{n}_0\sim\mathcal{CN}(\mathbf{0}, \sigma_0^2\mathbf{I}_{M})\in\mathbb{C}^{M\times1}$. The signal received at the fusion center is therefore given by:
\begin{align}
\mathbf{r}&=\sum_{i=1}^L\mathbf{H}_i\mathbf{F}_i\big(\mathbf{s}+\mathbf{n}_i\big)+\mathbf{n}_0 \label{eq:received_signal}\\
&=\Big(\sum_{i=1}^L\mathbf{H}_i\mathbf{F}_i\Big)\mathbf{s}+\Big(\underbrace{\sum_{i=1}^L\mathbf{H}_i\mathbf{F}_i\mathbf{n}_i+\mathbf{n}_0}_{\mathbf{n}}\Big),\label{eq:received_equivalent_signal}
\end{align}
where the compound noise vector $\mathbf{n}$ is still Gaussian: $\mathbf{n}\sim\mathcal{CN}(\mathbf{0},\mathbf{\Sigma}_{\mathbf{n}})$ whose covariance matrix $\mathbf{\Sigma}_{\mathbf{n}}$ is 
\begin{align}
\mathbf{\Sigma}_{\mathbf{n}}=\sigma_0^2\mathbf{I}_M+ \sum_{i=1}^L\mathbf{H}_i\mathbf{F}_i\mathbf{\Sigma}_i\mathbf{F}_i^H\mathbf{H}_i^H. \label{eq:cov_matrix_n}
\end{align}

It should be noted that the whiteness assumption of the Gaussian noise $\mathbf{n}_0$ at the receiver does not undermine the generality of the model. 
Indeed if $\mathbf{n}_0\sim\mathcal{CN}(\mathbf{0},\mathbf{\Sigma}_0)$ has coloured covariance $\mathbf{\Sigma}_0$, by redefining $\widetilde{\mathbf{r}}\triangleq \mathbf{\Sigma}_0^{-\frac{1}{2}}\mathbf{r}$, $\widetilde{\mathbf{H}}_i\triangleq\mathbf{\Sigma}_0^{-\frac{1}{2}}\mathbf{H}_i$ and $\widetilde{\mathbf{n}}_0\triangleq \mathbf{\Sigma}_0^{-\frac{1}{2}}\mathbf{n}_0$, the received signal can be equivalently written as  
\begin{align}
\widetilde{\mathbf{r}}=\sum_{i=1}^L\widetilde{\mathbf{H}}_i\mathbf{F}_i\big(\mathbf{s}+\mathbf{n}_i\big)+\widetilde{\mathbf{n}}_0,\label{eq:equivalent_signal}
\end{align}
with $\mathbf{\widetilde{\mathbf{n}}_0\sim\mathcal{CN}}(\mathbf{0}, \mathbf{I}_M)$, which falls into the model in (\ref{eq:received_signal}).


\subsection{Problem Formulation}
\label{subsec:problem_formulation}

In the beamforming problem that targets MMSE, a linear postcoder is usually employed on the destination end (i.e. the FC), and should be optimized together with the precoders at the transmitter end (i.e. sensors). 
In comparison, the definition of the mutual information does not include the receiver in the equation. This is due to the fundamental conclusion in information theory that mutual information can not increase whatever signal processing procedure is performed at the receiver. The mutual information between the source $\mathbf{s}$ and the reception at the FC is given in equation (\ref{eq:MI_full_definition}) shown in the next page:
\newcounter{TempEqnCnt}
\begin{figure*}[!t]
\normalsize
\setcounter{TempEqnCnt}{\value{equation}}
\setcounter{equation}{4}
\begin{align}
\mathsf{MI}\Big(\big\{\mathbf{F}_i\big\}_{i=1}^L\Big)=\log\det\bigg\{\mathbf{I}_M+\Big(\sum_{i=1}^L\mathbf{H}_i\mathbf{F}_i\Big)\mathbf{\Sigma}_{\mathbf{s}}\Big(\sum_{i=1}^L\mathbf{H}_i\mathbf{F}_i\Big)^H\Big(\sigma_0^2\mathbf{I}+\sum_{i=1}^L\mathbf{H}_i\mathbf{F}_i\mathbf{\Sigma}_i\mathbf{F}_i^H\mathbf{H}_i^H\Big)^{-1}\bigg\}
\label{eq:MI_full_definition}
\end{align}
\setcounter{equation}{\value{TempEqnCnt}}
\hrulefill
\vspace*{4pt}
\end{figure*}
\setcounter{equation}{5}

In practice, each sensor has independent transmit power constraint in accordance to individual battery supply. The average transmit power for the $i$-th sensor is $\mathsf{E}\big\{\big\|\mathbf{F}_i\big(\mathbf{s}+\mathbf{n}_i\big)\big\|^2_2\big\}=\Tr\big\{\mathbf{F}_i\big(\mathbf{\Sigma}_{\mathbf{s}}+\mathbf{\Sigma}_i\big)\mathbf{F}_i^H\big\}$, which must abide by its power constraint $P_i$. Thus the beamformer design problem of the noisy-observation sensing problem can be formulated as the following optimization problem:
\begin{subequations}
\label{eq:uniform_problem}
\begin{align}
\!\!\!\!\!\!\!(\mathsf{P}0)\!:\!\underset{\{\mathbf{F}_i\}_{i=1}^{L}}{\max.}& \ \mathsf{MI}\big(\{\mathbf{F}_i\}_{i=1}^L\big), \\
\!\!\!\!\!\!\!\mathrm{s.t.}& \Tr\big\{\mathbf{F}_i\big(\mathbf{\Sigma}_{\mathbf{s}}+\mathbf{\Sigma}_i\big)\mathbf{F}_i^H\big\}\!\leq\! P_i,\ i\!\in\!\{1,\cdots,L\}.
\end{align}
\end{subequations}

The above optimization problem is nonconvex, which can be easily shown by examining the convexity of the special case where $\{\mathbf{F}_i\}_{i=1}^L$ are all scalars. 
Convex problems are difficult, and cannot be solved in one shot. Below we tackle it using the popular BCA framework.


\section{Algorithm Design}
\label{sec:optimizing_MI}

\subsection{The General Approach via WMMSE and BCA}

The idea of BCA is to decompose the original problem of simultaneously optimizing the entire cohort of variables to multiple subproblems, each optimizing a subset of variables. However, for problem ($\mathsf{P0}$), directly applying the BCA method to partition the beamformers that await to be optimized into subgroups does not help simplify the problem. Even if we target only one beaformer at a time, the objective function is still very hard to maximize. Inspired by the weighted mean square error method proposed by the seminal papers \cite{bib:WMMSE_original} and \cite{bib:WMMSE_QingjiangShi}, we introduce auxiliary variables to convert the objective into a BCA-friendly form and subsequently decompose the problem into solvable subproblems. 

We first introduce two useful lemmas which pave the way for transforming the original hard problem ($\mathsf{P}0$).


\begin{lemma}[\cite{bib:WMMSE_original, bib:WMMSE_QingjiangShi}]
\label{lem:WMMSE_derivative}
For any positive definite matrix $\mathbf{E}\in\mathcal{H}^{n}_{++}$
, the following fact holds true
\begin{align}
\!\!\!\!\!\!\!\!-\log\det(\mathbf{E})=\underset{\mathbf{W}\in\mathcal{H}^{n}_{++}}{\max.}\big\{\ \log\det\big(\mathbf{W}\big)-\Tr\{\mathbf{W}\mathbf{E}\}+n\big\}
\end{align}
with the optimal solution $\mathbf{W}^{\star}$ given as
\begin{align}
\mathbf{W}^{\star}=\mathbf{E}^{-1}.
\end{align} 
\end{lemma}


\begin{lemma}
\label{lem:Wiener_filter_generalized}
(Generalized Wiener Filter) Define a matrix function $\mathbf{E}\big(\mathbf{G}\big)$ of variable $\mathbf{G}$ as
\begin{align}
\label{eq:matrix_function}
\mathbf{E}\big(\mathbf{G}\big)\triangleq\big(\mathbf{I}-\mathbf{G}^H\mathbf{H}\big)\mathbf{\Sigma}_{\mathbf{s}}\big(\mathbf{I}-\mathbf{G}^H\mathbf{H}\big)^H+\mathbf{G}^H\mathbf{\Sigma}_{\mathbf{n}}\mathbf{G},
\end{align}
where $\mathbf{\Sigma}_{\mathbf{s}}$ and $\mathbf{\Sigma}_{\mathbf{n}}$ are positive definite matrices. For any positive definite matrix $\mathbf{W}$, the following optimization problem 
\begin{align}
\underset{\mathbf{G}}{\min.}\Tr\big\{\mathbf{W}\mathbf{E}(\mathbf{G})\big\}\label{eq:opt_prob_trace_WE}
\end{align}
can be solved by the optimal solution 
\begin{align}
\mathbf{G}^{\star}=\big(\mathbf{H}\mathbf{\Sigma}_{\mathbf{s}}\mathbf{H}^H+\mathbf{\Sigma}_\mathbf{n}\big)^{-1}\mathbf{H}\mathbf{\Sigma}_{\mathbf{s}},\label{eq:MMSE_G}
\end{align} 
which leads to 
\begin{align}
\mathbf{E}(\mathbf{G}^{\star})=\big(\mathbf{H}^H\mathbf{\Sigma}_{\mathbf{n}}^{-1}\mathbf{H}+\mathbf{\Sigma}_{\mathbf{s}}^{-1}\big)^{-1}. \label{eq:minimal_MSE_matrix}
\end{align}
\end{lemma}

\begin{proof}
The problem in (\ref{eq:opt_prob_trace_WE}) is a convex problem. To see this, notice that the objective function in (\ref{eq:opt_prob_trace_WE}) is a quadratic function of $\mathbf{G}$ with its quadratic terms given as 
\begin{align}
\label{eq:quadratic_terms}
\Tr\big\{\mathbf{W}\mathbf{G}^H\mathbf{H}\mathbf{\Sigma}_{\mathbf{s}}\mathbf{H}^H\mathbf{G}\big\}+\Tr\big\{\mathbf{W}\mathbf{G}^H\mathbf{\Sigma}_{\mathbf{n}}\mathbf{G}\big\}.
\end{align}
By the identities $\mathsf{Tr}\{\mathbf{AB}\}=\mathsf{Tr}\{\mathbf{BA}\}$ and $\mathsf{Tr}\{\mathbf{ABCD}\}=vec^T(\mathbf{D}^T)\big[\mathbf{C}^T\otimes\mathbf{A}\big]vec(\mathbf{B})$, the first term of the above quadratic terms can be rewritten as 
\begin{align}
\!\!\!\!\!\!\!\!\!\!\!\!\!\Tr\big\{\mathbf{W}\mathbf{G}^H\mathbf{H}\mathbf{\Sigma}_{\mathbf{s}}\mathbf{H}^H\mathbf{G}\big\}\!\!=\!\!vec^H\!(\mathbf{G})\Big[\!\mathbf{W}^{\ast}\!\!\!\otimes\!\!\big(\mathbf{H}\mathbf{\Sigma}_{\mathbf{s}}\mathbf{H}^H\!\big)\!\Big]vec(\mathbf{G}).
\end{align}
Notice that $\mathbf{W}$ and $\mathbf{H}\mathbf{\Sigma}_{\mathbf{s}}\mathbf{H}^H$ are both positive semi-definite, so $\Big[\mathbf{W}^{\ast}\!\!\otimes\!\big(\mathbf{H}\mathbf{\Sigma}_{\mathbf{s}}\mathbf{H}^H\big)\Big]$ is positive semi-definite, 
and thus the first quadratic term is a convex function of $\mathbf{G}$. Similarly, the second quadratic term in (\ref{eq:quadratic_terms}) can also be proved to be convex function of $\mathbf{G}$. Thus (\ref{eq:opt_prob_trace_WE}) is a non-constrained convex problem of $\mathbf{G}$. By setting the derivative with respective to $\mathbf{G}$ to zero, 
we obtain 
\begin{align}
\!\!\frac{\partial\Tr\big\{\mathbf{W}\mathbf{E}(\mathbf{G})\big\}}{\partial\mathbf{G}^{\ast}}\!=\!\!\Big[\Big(\mathbf{H}\mathbf{\Sigma}_{\mathbf{s}}\mathbf{H}^H\!\!+\!\!\mathbf{\Sigma}_{\mathbf{n}}\Big)\mathbf{G}-\mathbf{H}\mathbf{\Sigma}_{\mathbf{s}}\Big]\mathbf{W}\!\!=\!\mathbf{O}.
\end{align}
Since $\mathbf{W}$ is positive definite, this clearly leads to 
$\Big(\mathbf{H}\mathbf{\Sigma}_{\mathbf{s}}\mathbf{H}^H\!\!+\!\!\mathbf{\Sigma}_{\mathbf{n}}\Big)\mathbf{G}^{\star}-\mathbf{H}\mathbf{\Sigma}_{\mathbf{s}}$, that is, the optimal solution in (\ref{eq:MMSE_G}).
\end{proof}

\begin{comment}
\label{cmt:1}
For the special case $\mathbf{W}=\mathbf{I}$, the conclusion in Lemma \ref{lem:Wiener_filter_generalized} results in the well-known Wiener filter. Hence Lemma \ref{lem:Wiener_filter_generalized} may be regarded as a generalized Wiener filtering result. The implication of this lemma is that, when the mean square error is weighted by a matrix $\mathbf{W}$, the Wiener filter preserves its optimality as long as the weight parameter $\mathbf{W}$ is positive definite.
\end{comment}

Given these two useful lemmas, we can now transform our objective function $\mathsf{MI}\big(\{\mathbf{F}_i\}_{i=1}^{L}\big)$ by introducing the notation
\begin{align}
\widetilde{\mathbf{H}}\triangleq\sum_{i=1}^L\mathbf{H}_i\mathbf{F}_i,
\end{align}
the notation in equation (\ref{eq:cov_matrix_n}), and a positive definite weight matrix $\mathbf{W}$:
\begin{align}
&\!\!\mathsf{MI}\Big(\{\mathbf{F}_i\}_{i=1}^L\Big)\!\!=\!\log\det\Big(\mathbf{I}_M+\widetilde{\mathbf{H}}\mathbf{\Sigma}_{\mathbf{s}}\widetilde{\mathbf{H}}^H\mathbf{\Sigma}_{\mathbf{n}}^{-1}\Big)\\
&\!\!\!\!\!\!\!\!\!\!=\!\log\det\Big(\big(\widetilde{\mathbf{H}}^H\mathbf{\Sigma}_{\mathbf{n}}^{-1}\widetilde{\mathbf{H}}\!+\!\mathbf{\Sigma}_{\mathbf{s}}^{-1}\big)\mathbf{\Sigma}_{\mathbf{s}}\Big) \\
&\!\!\!\!\!\!\!\!\!\!=\!-\!\log\det\Big(\widetilde{\mathbf{H}}^H\mathbf{\Sigma}_{\mathbf{n}}^{-1}\widetilde{\mathbf{H}}\!+\!\mathbf{\Sigma}_{\mathbf{s}}^{\!-\!1}\Big)^{\!\!-\!1}\!\!\!\!\!+\!\log\det\big(\mathbf{\Sigma}_{\mathbf{s}}\big) \\
&\!\!\!\!\!\!\!\!\!\!=\!\underset{\mathbf{W}\in\mathcal{H}^{K}_{++}}{\max.}\Big\{\log\det\big(\mathbf{W}\big)\!\!-\!\!\Tr\big\{\mathbf{W}\big(\widetilde{\mathbf{H}}^H\mathbf{\Sigma}_{\mathbf{n}}^{-1}\widetilde{\mathbf{H}}\!+\!\mathbf{\Sigma}_{\mathbf{s}}^{\!-\!1}\big)^{\!-\!1}\big\} \nonumber\\
&\qquad\qquad\!+\!\!K\Big\}\!+\!\log\det\big(\mathbf{\Sigma}_{\mathbf{s}}\big) \\
&\!\!\!\!\!\!\!\!\!\!=\!\underset{\stackrel{\mathbf{W}\in\mathcal{H}^{K}_{++},}{\mathbf{G}}}{\max.}\Big\{\log\det\big(\mathbf{W}\big)\!\!-\!\!\Tr\Big\{\mathbf{W}\big[\big(\mathbf{I}\!\!-\!\!\mathbf{G}^H\widetilde{\mathbf{H}}\big)\mathbf{\Sigma}_{\mathbf{s}}\big(\mathbf{I}\!\!-\!\!\mathbf{G}^H\widetilde{\mathbf{H}}\big)^H \nonumber\\
&\qquad\qquad\!+\!\mathbf{G}^H\mathbf{\Sigma}_{\mathbf{n}}\mathbf{G}\big]\Big\}\Big\}+\!\!K\!+\!\log\det\big(\mathbf{\Sigma}_{\mathbf{s}}\big),
\end{align}
where the last two steps follow from Lemma \ref{lem:WMMSE_derivative} and Lemma \ref{lem:Wiener_filter_generalized} respectively. 

We have thus transformed the MI optimization problem ($\mathsf{P}0$) to an equivalent problem ($\mathsf{P}1$) in (\ref{eq:opt_prob_MI_equivalent}) shown on the top of next page.
\begin{figure*}[!t]
\normalsize
\setcounter{TempEqnCnt}{\value{equation}}
\setcounter{equation}{21}
\begin{subequations}
\label{eq:opt_prob_MI_equivalent}
\begin{align}
(\mathsf{P}1)\underset{\stackrel{\mathbf{W}\in\mathcal{H}^{K}_{++},}{\{\mathbf{F}_i\}_{i=1}^L,\mathbf{G}}}{\max.}&\ \mathsf{MI}\Big(\!\{\!\mathbf{F}_i\!\}_{i=1}^{L},\!\mathbf{W},\!\mathbf{G}\!\Big)\!\label{eq:opt_obj_MI_equivalent}\\
&\!\!\!\!\!\!\!\!=\!\!\bigg\{\!\log\det\big(\!\mathbf{W}\!\big)\!\!-\!\!\Tr\bigg\{\!\mathbf{W}\!\Big[\!\Big(\mathbf{I}\!\!-\!\!\mathbf{G}^H\!\big(\sum_{i=1}^L\mathbf{H}_i\mathbf{F}_i\!\big)\!\Big)\mathbf{\Sigma}_{\mathbf{s}}\Big(\mathbf{I}\!\!-\!\!\mathbf{G}^H\!\big(\sum_{i=1}^L\mathbf{H}_i\mathbf{F}_i\!\big)\!\Big)^H\!\!\!+\!\!\mathbf{G}^H\mathbf{\Sigma}_{\mathbf{n}}\mathbf{G}\Big]\!\bigg\}\!\bigg\}\!\!+\!\!\log\det(\mathbf{\Sigma_{\mathbf{s}}})\!\!+\!\!K, \nonumber\\
&\mathrm{s.t.}\ \Tr\big\{\mathbf{F}_i\big(\mathbf{\Sigma}_{\mathbf{s}}\!\!+\!\!\mathbf{\Sigma}_{\mathbf{n}}\big)\mathbf{F}_i^H\big\}\!\leq\!P_i,\qquad\qquad i\in\{1,\cdots,L\}.\label{eq:opt_constr_MI_equivalent}
\end{align}
\end{subequations}
\setcounter{equation}{\value{TempEqnCnt}}
\hrulefill
\vspace*{4pt}
\end{figure*}
\setcounter{equation}{22}

We can now apply the idea of BCA and decompose the problem ($\mathsf{P}1$) into three subproblems that separately address the optimization of the set of precoders $\mathbf{F}_i$'s at the sensors, the weight matrix $\mathbf{W}$, and  the postcoder $\mathbf{G}$ at the FC. 
Applying the previous two lemmas, we readily obtain the closed-form optimal solutions to the latter two subproblems.


(i) When $\{\mathbf{F}_i\}_{i=1}^{L}$ and $\mathbf{G}$ are given, the optimal $\mathbf{W}^{\star}$ that achieves the maximal MI  in (\ref{eq:opt_prob_MI_equivalent}) is given by
\begin{align}
&\mathbf{W}^{\star}=\arg\underset{\mathbf{W}\in\mathcal{H}^{K}_{++}}{\max.}\mathsf{MI}\Big(\mathbf{W}\Big|\{\mathbf{F}_i\}_{i=1}^{L},\mathbf{G}\Big)\label{eq:optimal_W_MI}\\
&\!\!\!=\!\!\bigg[\Big(\mathbf{I}\!\!-\!\!\mathbf{G}^H\big(\sum_{i=1}^L\mathbf{H}_i\mathbf{F}_i\big)\Big)\mathbf{\Sigma}_{\mathbf{s}}\Big(\mathbf{I}\!\!-\!\!\mathbf{G}^H\big(\sum_{i=1}^L\mathbf{H}_i\mathbf{F}_i\big)\Big)^H\!\!\!\!+\!\mathbf{G}^H\mathbf{\Sigma}_{\mathbf{n}}\mathbf{G}\bigg]^{-1}.\nonumber
\end{align}

(ii) When $\{\mathbf{F}_i\}_{i=1}^{L}$ and $\mathbf{W}$ are given, the optimal $\mathbf{G}^{\star}$ that achieves the maximal MI  in (\ref{eq:opt_prob_MI_equivalent}) is given by
\begin{align}
\!\!&\mathbf{G}^{\star}=\arg\underset{\mathbf{G}}{\max.}\mathsf{MI}\Big(\mathbf{G}\Big|\{\mathbf{F}_i\}_{i=1}^{L},\mathbf{W}\Big)\label{eq:optimal_G_MI}\\
\!\!&\!\!=\!\Big[\big(\sum_{i=1}^L\mathbf{H}_i\mathbf{F}_i\big)\mathbf{\Sigma}_{\mathbf{s}}\big(\sum_{i=1}^L\mathbf{H}_i\mathbf{F}_i\big)^H\!\!+\!\!\mathbf{\Sigma}_{\mathbf{n}}\Big]^{-1}\!\!\Big(\sum_{i=1}^L\mathbf{H}_i\mathbf{F}_i\Big)\mathbf{\Sigma}_{\mathbf{s}},\nonumber
\end{align}
where $\mathbf{\Sigma}_{\mathbf{n}}$ is given in (\ref{eq:cov_matrix_n}).

We now address the first subproblem of optimizing $\{\mathbf{F}_i\}_{i=1}^{L}$ with $\mathbf{W}$ and $\mathbf{G}$ given. Towards this end, we have two options: we can jointly optimize $\{\mathbf{F}_i\}_{i=1}^L$ in one shot, or we can further apply BCA to partition the entire cohort of variables $\{\mathbf{F}_i\}_{i=1}^{L}$ into $L$ blocks, $\{\mathbf{F}_1\}$, $\cdots$, $\{\mathbf{F}_L\}$, and attack $L$ smaller problems one by one in a cyclic manner. For both options, solutions, hopefully in closed forms, are desirable and complexity is of critical importance.  Below we elaborate both approaches.


\subsection{Approach 1: Jointly Optimizing $\{\mathbf{F}_i\}_{i=1}^L$}
\label{subsec:3_block_BCD}

The subproblem of ($\mathsf{P1}$) maximizing $\mathsf{MI}\big(\{\mathbf{F}_i\}_{i=1}^L\big|\mathbf{W},\mathbf{G}\big)$ with $\mathbf{W}$ and $\mathbf{G}$ given is rewritten as follows
\begin{subequations}
\label{eq:opt_prob_P8}
\begin{align}
\!\!\!\!\!\!\!\!\!\!\!\!(\mathsf{P}2)\underset{\{\mathbf{F}_i\}_{i=1}^L}{\min.\!}&\!\Tr\bigg\{\!\mathbf{W}\Big[\Big(\mathbf{I}\!\!-\!\!\mathbf{G}^H\big(\sum_{i=1}^L\mathbf{H}_i\mathbf{F}_i\big)\Big)\mathbf{\Sigma}_{\mathbf{s}}\Big(\mathbf{I}\!\!-\!\!\mathbf{G}^H\big(\sum_{i=1}^L\mathbf{H}_i\mathbf{F}_i\big)\Big)^H\!\!\nonumber\\
&\qquad+\!\mathbf{G}^H\mathbf{\Sigma}_{\mathbf{n}}\mathbf{G}\Big]\bigg\}, \label{eq:opt_prob_P8}\\
\mathrm{s.t.}&\ \Tr\big\{\mathbf{F}_i\big(\mathbf{\Sigma}_{\mathbf{s}}\!\!+\!\!\mathbf{\Sigma}_{i}\big)\mathbf{F}_i^H\big\}\!\leq\!P_i,\ i\in\{1,\cdots,L\}.\label{eq:opt_constr_P8}
\end{align}
\end{subequations}

In this section, we consider solving ($\mathsf{P}2$) in a batch mode. 
We start by identifying the convexity of ($\mathsf{P}2$).


\begin{theorem}
\label{thm:convexity}
(Convexity)
The problem ($\it{P2}$) is convex.
\end{theorem}
\begin{proof}
To start, consider the following function $f\big(\mathbf{X}\big):\mathbb{C}^{m\times n}\mapsto\mathbb{R}$:
\begin{align}
f\big(\mathbf{X}\big)\triangleq\Tr\big\{\mathbf{\Sigma}_1\mathbf{X}\mathbf{\Sigma}_2\mathbf{X}^H\big\}
\label{eq:ffff}
\end{align}
with constant matrices $\mathbf{\Sigma}_1$ and $\mathbf{\Sigma}_2$ being positive semi-definite and having appropriate dimensions. By the identity $\mathsf{Tr}\{\mathbf{ABCD}\}=vec^T(\mathbf{D}^T)\big[\mathbf{C}^T\otimes\mathbf{A}\big]vec(\mathbf{B})$,  $f\big(\mathbf{X}\big)$ can be equivalently written as
\begin{align}
f\big(\mathbf{X}\big)&=vec^H\big(\mathbf{X})\big[\mathbf{\Sigma}_2^{\ast}\otimes\mathbf{\Sigma}_1\big]vec\big(\mathbf{X}\big).\label{eq:f_vecform}
\end{align}
Since $\mathbf{\Sigma}_1$ and $\mathbf{\Sigma}_2$ are positive semi-definite, $[\mathbf{\Sigma}_1^{\ast}\otimes\mathbf{\Sigma}_2]$ is positive semi-definite.
Thus $f\big(\mathbf{X}\big)$ is actually a convex function with respect to $\mathbf{X}$. 

Now, we replace $\mathbf{X}$ with $\mathbf{X}=\sum_{i=1}^L\mathbf{H}_i\mathbf{F}_i$. Since $\sum_{i=1}^L\mathbf{H}_i\mathbf{F}_i$ is an affine (actually linear) transformation of variables $\{\mathbf{F}_i\}_{i=1}^L$, and affine operations preserve convexity, 
the following function 
\begin{align}
f\big(\{\mathbf{F}_i\}_{i=1}^L\big)=\Tr\Big\{\mathbf{\Sigma}_1\big(\sum_{i=1}^L\mathbf{H}_i\mathbf{F}_i\big)\mathbf{\Sigma}_2\big(\sum_{i=1}^L\mathbf{H}_i\mathbf{F}_i\big)^H\Big\}
\end{align} 
is therefore convex  with respect to the set of variables $\{\mathbf{F}_i\}_{i=1}^L$. 

Finally, to identify the convexity of the objective in (\ref{eq:opt_prob_P8}), it suffices to prove that the nonlinear terms of $\{\mathbf{F}_i\}_{i=1}^L$ are also convex. The nonlinear terms are  given by
\begin{align}
&\Tr\Big\{\big(\mathbf{G}\mathbf{W}\mathbf{G}^H\big)\big(\sum_{i=1}^L\mathbf{H}_i\mathbf{F}_i\big)\mathbf{\Sigma}_\mathbf{s}\big(\sum_{i=1}^L\mathbf{H}_i\mathbf{F}_i\big)^H\Big\}\nonumber\\
&\qquad\qquad+\sum_{i=1}^L\Tr\Big\{\big(\mathbf{H}_i^H\mathbf{G}\mathbf{W}\mathbf{G}^H\mathbf{H}_i\big)\mathbf{F}_i\mathbf{\Sigma}_i\mathbf{F}_i^H\Big\}.
\end{align}
Again, based on the discussion at the beginning of this proof (i.e. (\ref{eq:ffff}) is convex), we can see that each of above terms is convex and thus the entire objective is convex. Following the same line of reasoning, each power constraint function is also convex. Hence the problem ($\mathsf{P}2$) is convex. 
\end{proof}

The identification of its convexity enables us to reformulate the problem ($\mathsf{P}2$) into a standard SOCP form. To proceed, we introduce the following notations
\begin{subequations}
\label{eq:notation_ABCD_1}
\begin{align}
\mathbf{f}_i&\triangleq vec\big(\mathbf{F}_i\big);\\
\mathbf{g}&\triangleq vec\big(\mathbf{G}\big);\\
\mathbf{A}_{ij}&\triangleq\mathbf{\Sigma}^*_{\mathbf{s}}\otimes\Big(\mathbf{H}^H_{i}\mathbf{G}\mathbf{W}\mathbf{G}^H\mathbf{H}_j\Big);\\
\mathbf{B}_i&\triangleq\big(\mathbf{W}\mathbf{\Sigma}_{\mathbf{s}}\big)^{\ast}\otimes\mathbf{H}_i;\\
\mathbf{C}_i&\triangleq\mathbf{\Sigma}^*_i\otimes\Big(\mathbf{H}_i^H\mathbf{G}\mathbf{W}\mathbf{G}^H\mathbf{H}_i\Big);\\
%
\mathbf{f}&\triangleq\big[\mathbf{f}_1^T,\cdots,\mathbf{f}_L^T\big]^T;\\
\mathbf{A}&\triangleq\big[\mathbf{A}_{ij}\big]_{i,j=1}^L;\\
\mathbf{B}&\triangleq\big[\mathbf{B}_1,\cdots,\mathbf{B}_L\big];\\
\mathbf{C}&\triangleq \mathsf{Diag}\big\{\mathbf{C}_1,\cdots,\mathbf{C}_L\big\};\\
&\!\!\!\!\!\mathbf{D}_i\!\!\triangleq\!\!\mathsf{Diag}\Big\{\!\mathbf{O}_{K(\!\sum_{j\!=\!1}^{i\!-\!1}\!\!N_j\!)}\!,\!\big(\mathbf{\Sigma}_{\mathbf{s}}\!\!+\!\!\mathbf{\Sigma}_i\big)^*\!\!\!\otimes\!\mathbf{I}_{N_i}\!,\!\mathbf{O}_{K(\!\sum_{j\!=\!i\!+\!1}^{L}\!\!N_j\!)}\!\!\Big\}; \label{eq:DD_i}\\
c&\triangleq \Tr\big\{\mathbf{W}\mathbf{\Sigma}_{\mathbf{s}}\big\}+\sigma_0^2\Tr\big\{\mathbf{G}\mathbf{W}\mathbf{G}^H\big\}. 
\end{align}
\end{subequations}
Based on these notations, problem ($\mathsf{P}2$) can be equivalently written as the following quadratically constrained quadratic problem (QCQP) problem:
\begin{subequations}
\label{eq:convex_QCQP}
\begin{align}
\!\!\!\!\!\!(\mathsf{P}3):\underset{\mathbf{f}}{\min}&\ 
\mathbf{f}^H\big(\mathbf{A}\!+\!\mathbf{C}\big)\mathbf{f}\!-\!2\Real\big\{\mathbf{g}^H\mathbf{B}\mathbf{f}\big\}\!+\!c, \\
s.t.& \ \ \mathbf{f}^H\mathbf{D}_i\mathbf{f}\leq P_i, \ \ \ \ \ \ i\in\{1,\cdots, L\}. 
\end{align}
\end{subequations}

By Theorem \ref{thm:convexity}, ($\mathsf{P3}$) is convex, and hence $(\mathbf{A}\!+\!\mathbf{C})$ is positive semidefinite, which implies that its square root $(\mathbf{A}\!+\!\mathbf{C})^{\frac{1}{2}}$ exists. Problem ($\mathsf{P}3$) can therefore be further simplified to a standard SOCP form as:
\begin{subequations}
\label{eq:convex_SOCP}
\begin{align}
&\!\!\!\!\!\!\!\!\!\!\!\!\!\!\!\!(\mathsf{P}3_{SOCP}): \underset{\mathbf{f},t,s}{\min.} \ t, \\
&\!\!\!\!\!\!\!\!\!\!\!\!\mathrm{s.t.}\ s-2\mathsf{Re}\{\mathbf{g}^H\mathbf{B}\mathbf{f}\}+c\leq t; \\
&
\left\| 
\begin{array}{c}
(\mathbf{A}\!+\!\mathbf{C})^{\frac{1}{2}}\mathbf{f} \\
\frac{s-1}{2}
\end{array}
\right\|_2\leq \frac{s\!+\!1}{2};\label{eq:SOC_1} \\
&\ \ \ \ \ \ 
\left\| 
\begin{array}{c}
\mathbf{D}_i^{\frac{1}{2}}\mathbf{f} \\
\frac{P_i-1}{2}
\end{array}
\right\|_2\leq \frac{P_i\!+\!1}{2},\ \ i\in\{1,\cdots,L\}\label{eq:SOC_i}.
\end{align}
\end{subequations}
The above problem can then be solved relatively efficiently by standard numerical tools like CVX \cite{bib:CVX}.  The entire approach is summarized in Algorithm \ref{alg:3_block}.


\begin{algorithm}
\caption{Batch-Mode BCA Algorithm to solve ($\mathsf{P}0$)}
\label{alg:3_block}
\textbf{Initialization}: randomly generate feasible $\{\mathbf{F}_i^{(0)}\}_{i=1}^{L}$; obtain $\mathbf{G}^{(0)}$ by (\ref{eq:optimal_G_MI}); obtain $\mathbf{W}^{(0)}$ by (\ref{eq:optimal_W_MI})\;
\Repeat{increase of $\mathsf{MI}$ is sufficiently small or predefined number of iterations is reached}
{
 with $\mathbf{G}^{(j-1)}$ and $\mathbf{W}^{(j-1)}$ being fixed, solve ($\mathsf{P}3$) in (\ref{eq:convex_SOCP}), obtain $\{\mathbf{F}_i^{(j)}\}_{i=1}^{L}$\;
 with $\{\mathbf{F}_i^{(j)}\}_{i=1}^{L}$ and $\mathbf{W}^{(j-1)}$ being fixed, obtain $\mathbf{G}^{(j)}$ by (\ref{eq:optimal_G_MI})\;
 with $\{\mathbf{F}_i^{(j)}\}_{i=1}^{L}$ and $\mathbf{G}^{(j)}$ being fixed, obtain $\mathbf{W}^{(j)}$ by (\ref{eq:optimal_W_MI})\;
}
\end{algorithm}

\begin{comment} 
\label{cmt:3_BCA}
We note that the idea of optimizing beamformers in a batch mode also appeared in a similar system model in \cite{bib:sensor_network_Cui} and \cite{bib:CWXing}, but our approach has its merit in terms of usability and/or efficiency. For example,   \cite{bib:sensor_network_Cui} obtained an approximated semidefinite SDP formulation that requires all the channel matrices be square and diagonal. Compared to \cite{bib:sensor_network_Cui}, our approach solves the exact problem ($\mathsf{P}1$), and is more general (applicable to arbitrary   channel matrices) and more efficient (SOCP has a lower complexity than SDP). The work of \cite{bib:CWXing} provided a very general framework for beamformer optimization, and presented the method of either SDP or SOCP to solve them. However, the specific SOCP formulation derived thereof, which takes a form similar to our formation ($\mathsf{P}3$), requires matrices $\big(\mathbf{A}\!+\!\mathbf{C}\big)$ and $\mathbf{D}_i$'s be positive definite and hence invertible. However, as we can see from (\ref{eq:DD_i}), $\mathbf{D}_i$'s are always non-invertible in our system setup (and there is no guarantee for $\big(\mathbf{A}\!+\!\mathbf{C}\big)$ to be invertible either). In comparison, the approach we developed in Algorithm \ref{alg:3_block} only requires $\mathbf{D}_i$'s and $\big(\mathbf{A}\!+\!\mathbf{C}\big)$ be positive semidefinite (which is always guaranteed). By relaxing the requirement, our approach appeals to a larger scope of scenarios than the SOCP approach in \cite{bib:CWXing}.   
\end{comment}


\subsection{Approach 2: Cyclically Optimizing $\mathbf{F}_i$'s}
\label{subsec:multiple_block}

The method discussed in the previous subsection that jointly optimizes all the $\mathbf{F}_i$'s has the advantage of guaranteed convergence, but the need to use some numerical solver (such as the interior point method) can cause a fairly high computational complexity, especially when the number of sensors and/or antennae is large (detailed complexity analysis is performed later in Subsection \ref{subsec:complexity}). 
It is highly desirable to develop closed-form solutions to ($\mathsf{P}3$), which unfortunately seems intractable.  

In what follows, we will resort to the BCA method again to further decompose ($\mathsf{P}0$) into $L$ sub-blocks, each optimizing a single precoder $\mathbf{F}_i$, for $i=1,\cdots, L$:
\begin{subequations}
\begin{align}
(\mathsf{P}3_i)\ \underset{\mathbf{f}_i}{\min.}&\ \mathbf{f}_i^H\big(\mathbf{A}_{ii}\!\!+\!\!\mathbf{C}_i\big)\mathbf{f}_i-2\mathsf{Re}\big\{\big(\mathbf{g}^H\mathbf{B}_i\!-\!\mathbf{q}_i^H\big)\mathbf{f}_i\big\},\\
\mathrm{s.t.}&\ \mathbf{f}_i^H\mathbf{E}_i\mathbf{f}_i\leq P_i
\end{align}
\end{subequations}
where $\mathbf{q}_i\triangleq\sum_{j\neq i}\mathbf{A}_{ij}\mathbf{f}_j$ and  
$\mathbf{E}_i\triangleq\big(\mathbf{\Sigma}_{\mathbf{s}}\!\!+\!\!\mathbf{\Sigma}_i\big)^*\!\!\otimes\!\mathbf{I}_{N_i}$.
The problem ($\mathsf{P}3_i$) is a convex trust region subproblem, and to solve it efficiently, we first introduce the following general result.

\begin{theorem}
\label{thm:sol_Cvx_TRSP}
(Convex Trust Region Subproblem)

Consider the following convex trust region subproblem given in ($\mathsf{P}_{Cvx-TRSP}$):
\begin{subequations}
\begin{align}
(\mathsf{P}_{Cvx-TRSP})\ \underset{\mathbf{x}_i\in\mathbb{C}^n}{\min.}&\ \mathbf{x}^H\mathbf{Q}\mathbf{x}-2\mathsf{Re}\big\{\mathbf{q}^H\mathbf{x}\big\},\\
\mathrm{s.t.}&\ \|\mathbf{x}\|_2 \leq \rho, \label{eq:cvx_TRSP_constr}
\end{align}
\end{subequations}
where $\mathbf{Q}\in\mathcal{H}_+^n$. Suppose that $\mathbf{Q}$ has rank $r$ and  eigenvalue decomposition $\mathbf{Q}=\mathbf{U}\mathbf{\Lambda}\mathbf{U}^H$ with $\mathbf{\Lambda}$ having its diagonal elements $\{\lambda_i\}_{i=1}^{n}$ arranged in a descending order. Further denote $\mathbf{p}\triangleq\mathbf{U}^H\mathbf{q}$, with $p_i$ being the $i$-th element of $\mathbf{p}$. 
Then the optimal solution to ($\mathsf{P}_{Cvx-TRSP}$) is given as follows. If $\mathbf{q}\perp\mathcal{N}(\mathbf{Q})$ and $\sum_{k\!=\!1}^{r}|p_k|^2\lambda_k^{-2}\!\leq\!\rho^2$, one optimal solution to ($\mathsf{P}_{Cvx-TRSP}$) is given by $\mathbf{x}^{\star}\!=\!\mathbf{Q}^{\dagger}\mathbf{q}$, which has the minimum $l_2$-norm among all possible optimal solutions; Otherwise,  the solution to ($\mathsf{P}_{Cvx-TRSP}$) is unique and given by $\mathbf{x}^{\star}\!=\!\big(\mathbf{Q}\!+\!\mu^{\star}\mathbf{I}_n\big)^{-1}\mathbf{q}$,  where $\mu^{\star}$ is the unique positive solution to 
the equation $\sum_{k=1}^{n}|p_k|^2(\lambda_k+\mu)^{-2}=\rho^2$.
\end{theorem}

\begin{proof}
Refer to Appendix \ref{subsec:appendix_thm_cvx_TRSP}.
\end{proof}

\begin{comment}
\label{cmt:sol_cvx_TRSP}
We have unveiled an important result in Theorem \ref{thm:sol_Cvx_TRSP}, which presents a significant enhancement of the standard solution to the general trust region subproblem, such as that presented in Section 4.3 of \cite{bib:NO}. Compared to the existing approaches in literature, where closed-form expression is missing and the solution is obtained through iterative numerical methods, the new result in Theorem \ref{thm:sol_Cvx_TRSP} has fully exploited the convexity of the objective, and is advantageous in two ways. First, the solution in Theorem \ref{thm:sol_Cvx_TRSP} has a much lower complexity, since it only requires a one-time eigenvalue decomposition with a complexity of $\mathcal{O}(n^3)$. 
Comparatively, standard methods (e.g. Algorithm 4.3 in \cite{bib:NO}) would run for many $\mathcal{O}(n^3)$) iterations, and each iteration needs one Cholesky decomposition and solving two linear systems, each of which takes a complexity $\mathcal{O}(n^3)$. Second, the new result clearly specifies the conditions for unique and multiple solutions, and in the presence of multiple solutions,  it gives out the one with the minimum $l_2$-norm,   which corresponds to the best power preserving solution. 
\end{comment}

Now return to our problem ($\mathsf{P}3_i$). By transforming the variable $\mathbf{f}_i$ into $\widetilde{\mathbf{f}}_i$ through $\mathbf{f}_i\!\triangleq\!\mathbf{E}_i^{-\!\frac{1}{2}}\widetilde{\mathbf{f}}_i$, and directly applying Theorem \ref{thm:sol_Cvx_TRSP}, we immediately get the solution to ($\mathsf{P}3_i$). 
A complete summary of this cyclic multi-block BCA algorithm is presented in Algorithm \ref{alg:multiple_block}.


The solution we just developed is highly efficient and works for general vector sources. In the special but important case of scalar sources (i.e. $K=1$), we can actually derive the fully analytical solution (without the need for  eigenvalue decomposition or bisection search). For notational consistency, in the scalar case, we rewrite the variables and parameters in problem ($\mathsf{P}3_i$) as
\begin{align}
\mathbf{W}&\rightarrow w;\ \mathbf{F}_i\rightarrow \mathbf{f}_i;\ \mathbf{G}\rightarrow \mathbf{g};\ \mathbf{\Sigma}_{\mathbf{s}}\rightarrow \sigma_s^2;\ \mathbf{\Sigma}_{i}\rightarrow \sigma_i^2.
\end{align}

By defining $\widetilde{\mathbf{q}}_i\triangleq \sum_{j\neq i}\mathbf{H}_j\mathbf{f}_j$, ignoring the terms independent of $\mathbf{f}_i$ and omitting the constant positive factor $w$ in the objective, we can rewrite the problem ($\mathsf{P}3_i$) as follows  
\begin{subequations}
\begin{align}
(\mathsf{P}4_i): \underset{\mathbf{f}_i}{\min.}&\ (\sigma_s^2+\sigma_i^2)\mathbf{f}_i^H\big(\mathbf{H}_i^H\mathbf{g}\mathbf{g}^H\mathbf{H}_i\big)\mathbf{f}_i\nonumber\\
&\quad\qquad-2\sigma_s^2\mathsf{Re}\big\{(1-\widetilde{\mathbf{q}}_i^H\mathbf{g})\mathbf{g}^H\mathbf{H}_i\mathbf{f}_i\big\} \\
\mathrm{s.t.}&\ \|\mathbf{f}_i\|^2\leq \frac{P_i}{\sigma_s^2+\sigma_i^2}\triangleq\bar{P}_i.
\end{align}
\end{subequations}


\begin{theorem}
\label{thm:scalar_case}
(Precoder $\mathbf{F}_i$ for Scalar Sources)
For scalar transmission ($\mathit{K=1}$), fully analytical optimal solution to problem ($\mathit{P4}_i$) can be obtained without invoking bisection search or eigenvalue decomposition as follows:

If $\sigma_s^4\big|1-\mathbf{g}^H\widetilde{\mathbf{q}}_i\big|^2>(\sigma_s^2+\sigma_i^2)^2\bar{P}_i\|\mathbf{H}_i^H\mathbf{g}\|^2_2$,
\begin{subequations}
\begin{align}
\!\!\!\!\!\!\!\!\mu_i^{\star}&\!=\!\sigma_s^2\bar{P}_i^{-\frac{1}{2}}\big|1\!-\!\mathbf{g}^H\widetilde{\mathbf{q}}_i\big|\big\|\mathbf{H}_i^H\mathbf{g}\big\|_2\!-\!(\sigma_s^2\!+\!\sigma_i^2)\big\|\mathbf{H}_i^H\mathbf{g}\big\|_2^2, \label{eq:opt_positive_mu_i_scalar_case}\\
\!\!\!\!\!\!\!\!\mathbf{f}_i^{\star}&\!=\!\sigma_s^2(1\!-\!\mathbf{g}^H\widetilde{\mathbf{q}}_i)\Big(\mu_i^{\star}\mathbf{I}\!+\!(\sigma_s^2\!+\!\sigma_i^2)\mathbf{H}_i^H\mathbf{g}\mathbf{g}^H\mathbf{H}_i\Big)^{-1}\!\!\mathbf{H}_i^H\mathbf{g}.
\end{align}
\end{subequations}
\indent Otherwise 
\begin{align}
\!\!\!\!\mathbf{f}_i^{\star}\!=\!\frac{\sigma_s^2(1\!-\!\mathbf{g}^H\widetilde{\mathbf{q}}_i)\mathbf{H}_i^H\mathbf{g}}{(\sigma_s^2\!+\!\sigma_i^2)\mathbf{g}^H\mathbf{H}_i\mathbf{H}_i^H\mathbf{g}}.
\end{align}
\end{theorem}

\begin{proof}
Refer to Appendix \ref{subsec:appendix_thm_scalar_case}.
\end{proof}


\begin{algorithm}
\caption{Cyclic Multi-block BCA Algorithm to Solve ($\mathsf{P}0$)}
\label{alg:multiple_block}
\textbf{Initialization}: randomly generate feasible $\{\mathbf{F}_i^{(0)}\}_{i=1}^{L}$; 
obtain $\mathbf{G}^{(0)}$ by (\ref{eq:optimal_G_MI}); obtain $\mathbf{W}^{(0)}$ by (\ref{eq:optimal_W_MI})\;
\Repeat{increase of $\mathsf{MI}$ is sufficiently small or predefined number of iterations is reached}
{ 
		\For{$i=1;\ i<=L;\ i++$}
		{ 
		\uIf{$K=1$}
		{   Utilize Theorem \ref{thm:scalar_case} to obtain $\mathbf{F}_i$\;
		}
		  \Else
		{   
		    Utilize Theorem \ref{thm:sol_Cvx_TRSP} to obtain $\mathbf{F}_i$\;
		}

		  update $\mathbf{G}^{}$ by (\ref{eq:optimal_G_MI}) \;
		  update $\mathbf{W}^{}$ by (\ref{eq:optimal_W_MI}) \; 
		}
	}

\end{algorithm}

\begin{comment}
\label{cmt:cyclic_BCA}
Two papers on MIMO AF relay systems are worth comparing to, as they addressed similar problems in a cyclic multi-block manner. The work of \cite{bib:MIMO_AF_Relay} also used the WMMSE-based framework, but the method proposed to solve the individual blocks is based on semidefinite relaxation (which involves relaxing the original subproblem, solving the relaxed subproblem via an SDP numerical solver, and then iteratively reducing the rank to obtain the solution to the original subproblem), and therefore has a much higher complexity than ours (see complexity comparison in Table \ref{tab:runningtime1}). The method presented in \cite{bib:TSTINR} converted the original MI maximization problem to a different problem of maximizing the total-signal-to-total-interference-plus-noise-ratio (TSTINR). This TSTINR model can then be decomposed and converted to a eigenvalue-decomposition and trust-region subproblem. The primary advantage advocated for TSTINR is that, despite its disappointing performance in the low signal-to-noise (SNR) region, it may potentially outperform WMMSE-based approaches in the high SNR region.  
However, for the system model evaluated here, we have not seen any performance advantage of the TSTINR approach: it either falls short of, or, at the best performs on par with our proposed approaches (see Figure \ref{fig:comparisonTSTINR}).    
\end{comment}



\subsection{Convergence}
\label{subsec:convergence}

For the first batch-mode BCA algorithm developed previously, we also have a quite strong convergence result given in Theorem (\ref{thm:convergence_3BCA}), whose proof is detailed in the Appendix.


\begin{theorem}
\label{thm:convergence_3BCA} (Convergence) 
Suppose that the covariance matrix $\mathbf{\Sigma}_{\mathbf{s}}\succ0$ (i.e. signal has bounded energy). Algorithm \ref{alg:3_block} generates an increasing MI sequence. Its solution sequence has limit points, and each limit point of the solution sequence is a Karush-Kuhn-Tucker (KKT) point of the original problem ($\mathsf{P}0$).
\end{theorem}

\begin{proof}
Refer to appendix \ref{subsec:appendix_thm_convergnece_3BCA}. 
\end{proof}

\begin{comment}
\label{cmt:convergence}
Convergence property is very important for BCA/BCD approaches. There exists example showing that inappropriate BCA/BCD algorithm can have its iterates always far from optimality and thus result in bad solutions. Convergence proof is usually challenging, and previous papers were able to make conjectures but not proofs (e.g. \cite{bib:MIMO_AF_Relay} \cite{bib:CWXing}). Theorem \ref{thm:convergence_3BCA} shows that by judiciously transforming the problem, we have developed an approach (Algorithm 1) to solve the hard nonconvex problem ($\mathsf{P}0$) which not only guarantees the convergence of the objective, but also the convergence of the solution. Note that although most of the BCA/BCD-based approaches in the beamformer literature expect a converging objective, convergence proof of the solution sequence is usually missing \cite{bib:Sensor_compress_1, bib:Sensor_compress_2, bib:FangLi_2, bib:Sensor_compress_3, bib:YangLiu_ICC, bib:sensor_network_Cui, bib:sensor_network_Hamid, bib:JunFang_MI, bib:YangLiu_ISIT14, bib:MIMO_AF_Relay, bib:CWXing, bib:TSTINR}.  
\end{comment}

For the cyclic multiple block approach, we are unable to analytically prove the convergence. However, extensive simulations demonstrate that they also converge, and, in some cases, takes fewer iterations to converge than the batch-mode approach (not to mention that each iteration requires considerably less complexity). 
  
\subsection{Complexity}
\label{subsec:complexity}

Complexity and speed is a critical performance indicator for nonconvex optimization. The original beamformer problem is decomposed to three subproblems. 

(i) The first two subproblems of optimizing $\mathbf{G}$ and $\mathbf{W}$ have closed-form solutions given in (\ref{eq:optimal_G_MI}) and (\ref{eq:optimal_W_MI}), and their complexities come primarily from matrix inversion which is $\mathcal{O}\big(K^3\big)$. 

(ii) For the third subproblem of optimizing $\{\mathbf{F}_i\}$, we have developed a batch-mode approach and a cyclic multi-block approach. 

The batch-mode BCA algorithm (Alg. \ref{alg:3_block}): Problem ($\mathsf{P}2$) is equivalently transformed to the SOCP problem ($\mathsf{P}3_{SOCP}$) formulated in (\ref{eq:convex_SOCP}) and, according to the SOCP complexity  listed in \cite{bib:IPM_Terlaky}, it has a complexity of $\mathcal{O}\Big(\sqrt{L}K^3(\sum_{i=1}^LN_i)^3\Big)$. This is also the complexity for each loop of the batch-mode BCA algorithm. 


The cyclic multi-block BCA algorithm (Alg. \ref{alg:multiple_block}): Problem ($\mathsf{P}3_i$) optimizing one beamformer $\mathbf{F}_i$ has its major complexity coming from the eigenvalue decomposition, which is $\mathcal{O}(K^3N_i^3)$. Thus the complexity for each loop is $\mathcal{O}\Big(K^3\sum_{i=1}^LN_i^3\Big)$. (The complexity here is considerably lower than  that of the 3-block BCA algorithm.)

The only two other methods that may be used to solve our problem, the WMMSE-SDR algorithm \cite{bib:MIMO_AF_Relay} and the TSTINR algorithm \cite{bib:TSTINR}, both operate in the multi-block mode and both have higher complexity. Since they both utilize semidefinite relaxation to update each beamformer, it can be shown from \cite{bib:ModernCVX} that the complexity for updating each beamformer is $\mathcal{O}(K^{3.5}N_i^{3.5})$. Thus the complexity for each outer loop, which sequentially updates every beamformer, is $\mathcal{O}\big(K^{3.5}(\sum_{i=1}^LN_i^{3.5})\big)$.   


To summarize, the proposed cyclic multi-block BCA approach  has considerably lower complexity than all of its peers. The proposed batch-mode BCA approach is efficient when the number of sensors ($L$) is modest\footnote{A particularly meaningful attribute of the batch-mode algorithm is that its convergence can be rigorously proven.}. These  complexity results are also verified via simulations. 


\section{Numerical Results}
\label{sec:numerical results}

We now provide simulation results to verify the effectiveness of our proposed solutions. 
Accounting for the possible correlation between antennae (which is related to their mutual distances), we set the covariance matrices of the source signal and of all the sensing noise as 
\begin{align}
\mathbf{\Sigma}_{\mathbf{s}}=\sigma_s^2\mathbf{\Sigma}_0, \ \ \  \mathbf{\Sigma}_i=\sigma_i^2\mathbf{\Sigma}_0,\ \ \ i\in\{1,\cdots,L\}, 
\end{align} 
where the $K\times K$ Toeplitz matrix $\mathbf{\Sigma}_0$ is given as $\big[\mathbf{\Sigma}_{0}\big]_{j,k}=\gamma^{|j-k|}$. 
The parameter $\gamma$ here is used to adjust the correlation level between different components of the signal or noise. In our experiments, $\gamma$ is set to $\gamma=0.5$. We define the sensing signal-to-noise-ratio (SNR) at the $i$-th sensor as $\mathsf{SNR}_i\triangleq\frac{\sigma_s^2}{\sigma_i^2}$ and the channel SNR as $\mathsf{SNR}\triangleq\frac{\sigma_s^2}{\sigma_0^2}$. For the implementation of the batch-mode BCA algorithm and the algorithms in \cite{bib:MIMO_AF_Relay} \cite{bib:TSTINR}, SDPT3 solver of CVX is used.

Figures \ref{fig:Gen_MI_Alg4Alg5_Diff_Iterations_test1} and \ref{fig:Gen_MI_Alg4Alg5_Diff_Iterations_test2} demonstrate the performances of the batch-mode BCA and the cyclic multi-block BCA for vector source signals. For comparison purpose, the WMMSE-SDR algorithm in \cite{bib:MIMO_AF_Relay} and the (average performance) of random beamformers with full-power-transmission (block dotted curve) are also plotted. Two network profiles are tested: a heterogeneous network where each sensor is different in term of the transmission power, the observation noise level and the numbers of antennae, and a homogeneous network where all the sensors are identically configured. 
To take into account the impact of channel parameters, at each specific channel SNR, we randomly generate 300 different channel realizations, and for each channel realization, our algorithms and the one in \cite{bib:MIMO_AF_Relay} are performed, all starting from the same initial solution that is randomly generated and normalized to satisfy the power constraints. The average mutual information (averaged over 300 channel realizations) with respect to the (outer-loop) iteration is plotted. In both network types, we see that the optimized beamformers perform tremendously better than the non-optimized random beamformers. We see that the algorithm in \cite{bib:MIMO_AF_Relay} has a per-iteration convergence almost identical to our cyclic multi-block BCA (but each iteration thereof has a much higher complexity)\footnote{The algorithm in \cite{bib:MIMO_AF_Relay} is also multi-block BCA based. The difference and advantage of Algorithm 2 proposed here is the derivation of near closely form solutions rather than that based on  numerical solvers.}. In all the cases, 40 to 50 iterations are sufficient to make the results converge and all algorithms deliver about the same amount of mutual information.

\begin{figure}
\centering
\includegraphics[height=2.8in,width=3.6in]{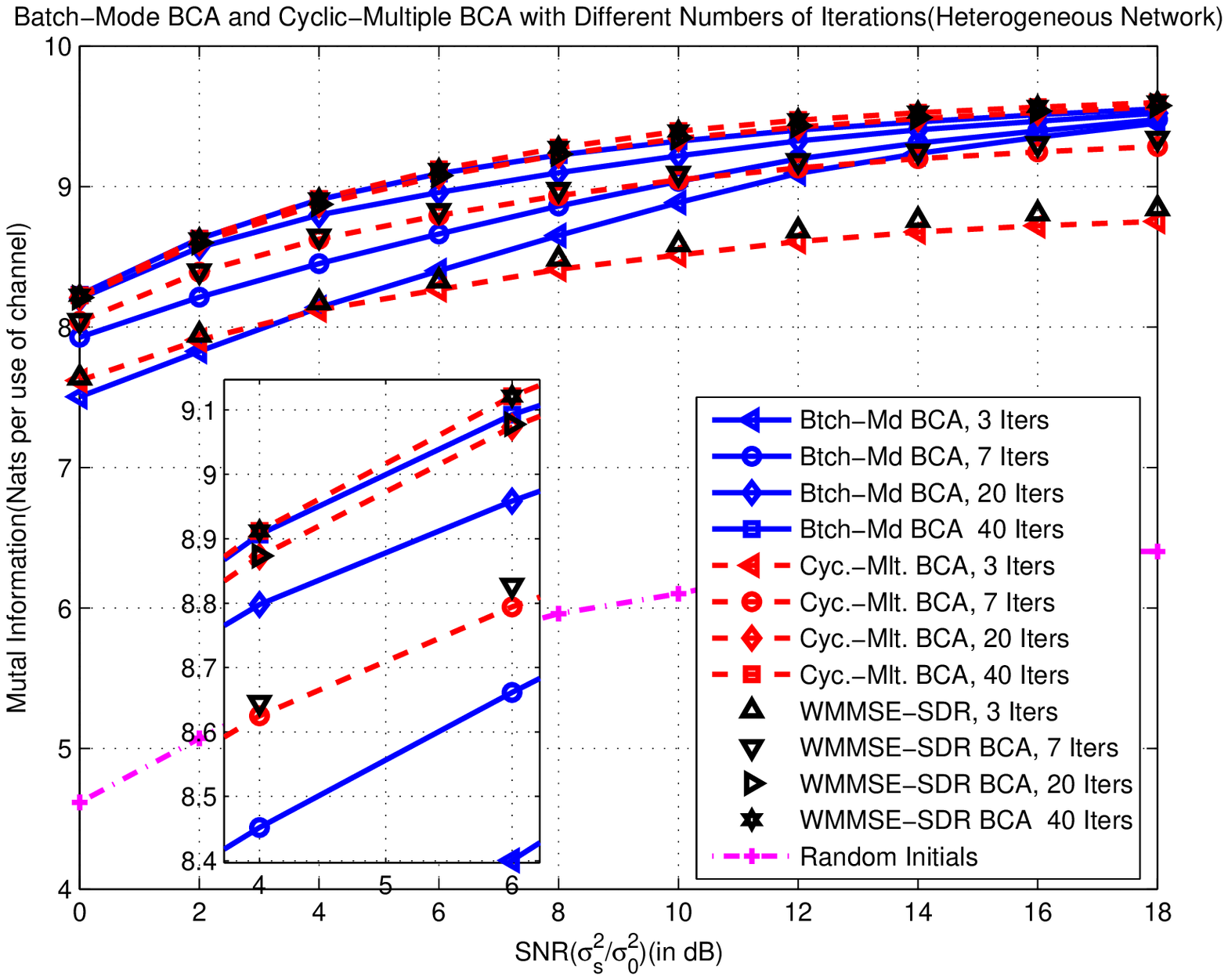}
\caption{Heterogeneous Network: batch-mode BCA and cyclic multi-block BCA. Source dimension $K\!=\!3$; Number of sensors $L\!=\!3$. Number of antennae at sensors and FC: $N_1\!=\!3$, $N_2\!=\!4$, $N_3\!=\!5$, $M\!=\!4$. Sensing SNR and transmit power constraints at sensors: 
$\mathsf{SNR}_1\!=\!8dB$, $\mathsf{SNR}_2\!=\!9dB$, $\mathsf{SNR}_3\!=\!10dB$,
$P_1\!=\!2$, $P_2\!=\!2$, $P_3\!=\!3$. The performance of WMMSE-SDR \cite{bib:MIMO_AF_Relay} virtually overlaps with our cyclic multi-block BCA.}
\label{fig:Gen_MI_Alg4Alg5_Diff_Iterations_test1} 
\end{figure}

\begin{figure}
\centering
\includegraphics[height=2.8in,width=3.6in]{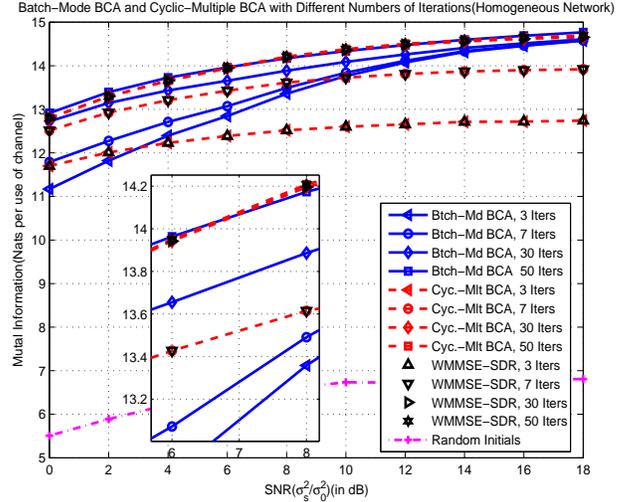}
\caption{Homogeneous Network: batch-mode BCA and cyclic multi-block BCA. Source dimension $K\!=\!4$; Number of sensors $L\!=\!5$. Number of antennae at sensors and FC: $N_i\!=\!5$, $M\!=\!4$. Sensing SNR and transmit power constraints at sensors: 
$\mathsf{SNR}_i\!=\!9dB$, $P_i\!=\!2$. The performance of WMMSE-SDR \cite{bib:MIMO_AF_Relay} virtually overlaps with our cyclic multi-block BCA.}
\label{fig:Gen_MI_Alg4Alg5_Diff_Iterations_test2} 
\end{figure}

Figure \ref{fig:Scalar_MI_Alg4Alg5_Diff_Iterations} tests the special case of scalar source signal ($K\!=\!1$), where the multi-block BCA algorithm has fully closed-form solution given in Theorem \ref{thm:scalar_case}. The observation is largely similar to that in the multi-dimension source case, but there is also subtle difference. Here we see that across the entire region of channel SNRs, the multi-block BCA algorithm appears to converge faster with the number of iterations (and each iteration is super fast with all the closed-form solutions). At low SNRs, it appears 10 and 30 iterations are good for cyclic multi-block BCA and batch-mode BCA respectively, and at high SNRs, less than 10 iterations are needed for both algorithms. The other interesting observation here is that even random beamformers can conveniently achieve the maximum MI provided that the channel SNR is sufficiently large.

\begin{figure}
\centering
\includegraphics[height=2.8in,width=3.6in]{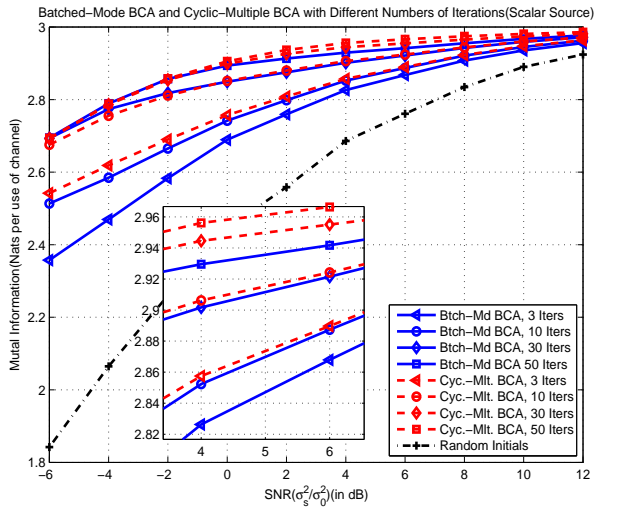}
\caption{Scalar Source ($K\!=\!1$) in a heterogeneous network: batch-mode BCA and cyclic multi-block BCA. Same configuration as in Fig. \ref{fig:Gen_MI_Alg4Alg5_Diff_Iterations_test1} except $K\!=\!1$.
}
\label{fig:Scalar_MI_Alg4Alg5_Diff_Iterations} 
\end{figure}

Figures \ref{fig:Gen_MI_Alg4Alg5_Diff_Initials_test1} and \ref{fig:Gen_MI_Alg4Alg5_Diff_Initials_test2} evaluate the impact of the random initials used to start the iterative BCA optimization for heterogeneous and homogeneous networks. 
We have tested a large number of random initials, each satisfying the power constraint with equality, and the observation is quiet consistent. To avoid messy plots, only 10 random initials are shown (the same 10 initials for both algorithms). We see that both of our algorithms are rather robust, exhibiting insensitivity to the  initial points, and capable of converging to identical performance.  

\begin{figure}
\centering
\includegraphics[height=2.8in,width=3.6in]{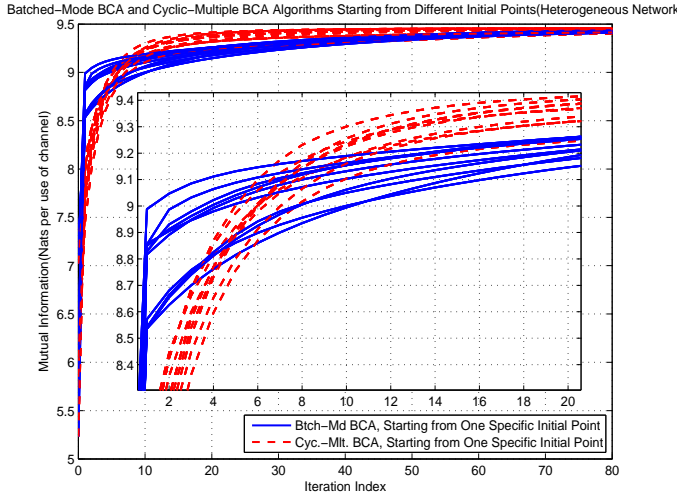}
\caption{Heterogeneous network : Impact of different random initial points. Same configuration as in Fig. \ref{fig:Gen_MI_Alg4Alg5_Diff_Iterations_test1}. Channel $\mathsf{SNR}\!=\!8dB$.}
\label{fig:Gen_MI_Alg4Alg5_Diff_Initials_test1} 
\end{figure}

\begin{figure}
\centering
\includegraphics[height=2.8in,width=3.6in]{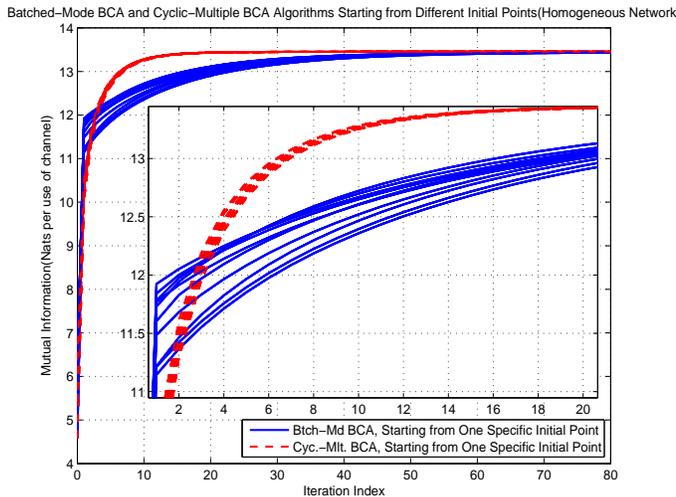}
\caption{
Homogeneous network: Impact of different random initial points. Channel $\mathsf{SNR}\!=\!4dB$. Same configuration as in Fig. \ref{fig:Gen_MI_Alg4Alg5_Diff_Iterations_test2}.}
\label{fig:Gen_MI_Alg4Alg5_Diff_Initials_test2} 
\end{figure}

We further compare our proposed algorithms with the two recent algorithms developed in \cite{bib:MIMO_AF_Relay}  and  \cite{bib:TSTINR}. Putting aside the differences in complexity (shown in Table \ref{tab:runningtime1}), our new algorithms and the WMMSE-based algorithm developed in \cite{bib:MIMO_AF_Relay} exhibit essentially the same performance (after 50 iterations). For the TSTINR algorithm developed in \cite{bib:TSTINR}, although it was mentioned in  \cite{bib:TSTINR} that it could be advantageous at high SNR, we have not observed it in our system model. As shown in Figure \ref{fig:comparisonTSTINR}, in both homogeneous and heterogeneous system setups, our proposed methods deliver consistently good performance for the entire channel SNR range. The TSTINR method falls far behind (even worse than random beamformers) at low SNRs, starts to catch up as channel quality improves, and eventually performs on par with our methods (rather than exceed them) at high SNRs.

\begin{figure}
\centering
\includegraphics[height=2.8in,width=3.6in]{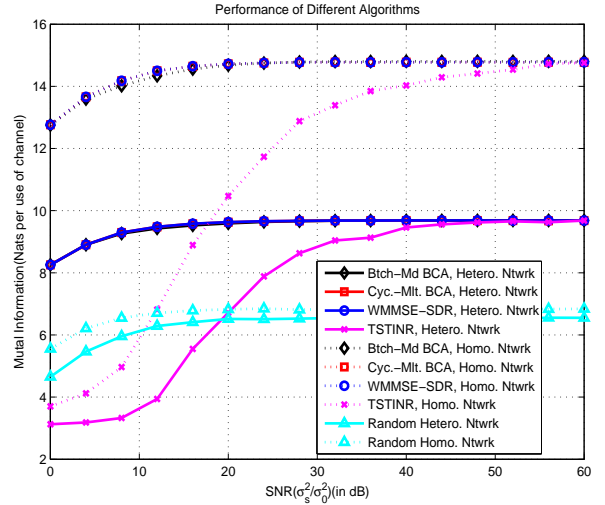}
\caption{
Performance: Our Algorithms vs. the WMMSE-SDR in \cite{bib:MIMO_AF_Relay} v.s. the TSTINR Algorithm in \cite{bib:TSTINR}}
\label{fig:comparisonTSTINR} 
\end{figure}

Finally, as an indication to complexity, we evaluate in Table \ref{tab:runningtime1} the average Matlab run-time of our algorithms and the existing ones \cite{bib:MIMO_AF_Relay} and \cite{bib:TSTINR}. A homogeneous sensing network is considered with different configurations of $K$, $L$ and $N_i$, which result in different sizes of the problem. SDPT3 solver of the CVX is used as the numerical solver when needed. As expected, the near closed-form solutions enable our cyclic mutli-block BCA algorithm to run extremely fast, taking only a tiny fraction of run-time compared to the existing approaches  \cite{bib:MIMO_AF_Relay}\cite{bib:TSTINR}. For relatively small networks, i.e. $K$ and/or $L$ is small, our batch-mode BCA algorithm (with guaranteed convergence) is also quite efficient.

\begin{table}
\caption{\small{MATLAB Run-Time Per Outer-Loop (in Sec.)}}
\centering
\begin{tabular}{|c|c|c|c|c|}
\hline
\backslashbox{Dim.}{$L$}& Algorithms &$L=5$&$L=10$&$L=20$ \\
\hline
$K=1$ & our Btch-Md & 0.2655 & 0.3956 & 0.7675  \\
\cline{2-5}
$M=2$ & our Cyc. Mlt. & 0.0125 & 0.0442 & 0.1765  \\
\cline{2-5}
$N_i=2$ & \scriptsize{WMMSE-SDR} \cite{bib:MIMO_AF_Relay} & 0.7037 & 1.475 & 3.244  \\
\cline{2-5}
        & TSTINR \cite{bib:TSTINR} & 0.7295 & 1.479 & 3.269  \\
\hline
$K=4$ & our Btch-Md & 0.5137 & 1.914 & 13.88  \\
\cline{2-5}
$M=4$ & our Cyc. Mlt. & 0.0328 & 0.1056 & 0.3818  \\
\cline{2-5}
$N_i=4$ & \scriptsize{WMMSE-SDR} \cite{bib:MIMO_AF_Relay} & 1.139 & 2.293 & 4.883  \\
\cline{2-5}
        & TSTINR \cite{bib:TSTINR} & 1.178 & 2.463 & 5.325  \\
\hline
$K=8$ & our Btch-Md & 10.01 & 80.98 & 668.0  \\
\cline{2-5}
$M=8$ & our Cyc. Mlt. & 0.0496 & 0.1431 & 0.4764  \\
\cline{2-5}
$N_i=8$ & \scriptsize{WMMSE-SDR} \cite{bib:MIMO_AF_Relay} & 5.217 & 10.31 & 20.73  \\
\cline{2-5}
        & TSTINR \cite{bib:TSTINR} & 4.951 & 10.17 & 20.71  \\
\hline
\end{tabular}\\
\vspace{0.03in}
Notes: SDPT3 solver of CVX is chosen to implement the algorithms in \cite{bib:MIMO_AF_Relay}, \cite{bib:TSTINR} and our batch-mode BCA. \\
\label{tab:runningtime1}
\end{table}

\section{Conclusion}
\label{sec:conclusion}

We have developed two new and efficient methods to optimize linear beamformers for a coherent MAC wireless sensor network with noisy observation and multiple MIMO sensors. 
Targeting maximizing the end-to-end mutual information, and leveraging the WMMSE and the BCA framework, we have developed a batch-mode approach and a cyclic multi-block approach.
Our batch-mode approach uses an SOCP formulation (which lends itself to relatively efficient numerical solvers), has a guaranteed convergence that we have rigorously proven, and, compared with the existing match-mode approaches, does not impose such stringent requirement of matrices being positive-definiteness.
Our cyclic multi-block approach further allows us to derive a nearly closed-form solution (up to a quick bisection search) to each block and hence renders extremely low complexity.  Extensive simulations are presented to demonstrate the effectiveness of the new approaches.


\appendix

\subsection{Proof of Theorem \ref{thm:sol_Cvx_TRSP}}
\label{subsec:appendix_thm_cvx_TRSP}

\begin{proof}
The problem ($\mathsf{P}_{Cvx-TRSP}$) is convex and Slater's condition is obviously satisfied. Thus it suffices to find the solution to its KKT conditions \cite{bib:CvxOpt}. By equivalently writing the constraint (\ref{eq:cvx_TRSP_constr}) in a squared-norm form and introducing its Lagrangian multiplier $\mu$, the KKT conditions are given as  
\begin{subequations}
\begin{align}
\big(\mathbf{Q}+\mu\mathbf{I}_n\big)\mathbf{x}
&=\mathbf{q};\label{eq:KKT_a}\\
\|\mathbf{x}\|^2_2&\leq \rho^2;\label{eq:KKT_b}\\
\mu\big(\|\mathbf{x}\|^2_2-\rho^2\big)&=0;\label{eq:KKT_c}\\
\mu &\geq0. \label{eq:KKT_d}
\end{align}\label{eq:KKT}
\end{subequations}

Let $f(\mu)\triangleq\sum_{k=1}^{n}|p_k|^2(\lambda_k+\mu)^{-2}$. 

First consider the case $\mu>0$. By the KKT conditions (\ref{eq:KKT_a}) and (\ref{eq:KKT_c}), we have
\begin{align}
\!\!\!\!\!\!\|\mathbf{x}\|_2^2\!&=\big\|\!\big(\mathbf{Q}\!+\!\mu\mathbf{I}_n\big)^{-1}\mathbf{q}\big\|_2^2\!\\
&=\!\mathbf{q}^H\!\mathbf{U}\big(\mathbf{\Lambda}\!+\!\mu\mathbf{I}_n\big)^{\!-\!2}\mathbf{U}^H\!\mathbf{q}\\
&=\sum_{k=1}^n|p_k|^2(\lambda_k\!+\!\mu)^{-2}\\
&=f(\mu)=\!\rho^2. \label{eq:KKT_analysis_1}
\end{align}
Since $f(\mu)$ is strictly decreasing with respect to $\mu$, in order for (\ref{eq:KKT_analysis_1}) to have a positive solution, the supremum of $f(\mu)$,  $\underset{\mu\rightarrow0}{\lim}f(\mu)$, must be larger than $\rho^2$. Only two specific cases will make this condition true: the first is $\mathbf{q}\not\perp\mathcal{N}(\mathbf{Q})$, which means $\exists i\in\{r\!+\!1,\cdots,n\}$ such that $p_i\neq0$; and the second is $\mathbf{q}\perp\mathcal{N}(\mathbf{Q})$ and $f(0)\!=\!\sum_{k=1}^r|p_k|^2\lambda_k^{-2}>\rho^2$ (where $r$ is the rank of $\mathbf{Q}$).  If either of the aforementioned cases holds, there exists a positive $\mu^{\star}$ satisfying (\ref{eq:KKT_analysis_1}) (which can be efficiently found via bisection search), and the unique solution to ($\mathsf{P}_{Cvx-TRSP}$) can therefore be derived as $\mathbf{x}^{\star}=\big(\mathbf{Q}\!+\!\mu\mathbf{I}_n)\big)^{-1}\mathbf{q}$. 

Next consider the case $\mu=0$. From the above discussion, at this time, $\mathbf{q}\perp\mathcal{N}(\mathbf{Q})$ and $\sum_{k=1}^r|p_k|^2\lambda_k^{-2}\leq\rho^2$.
From (\ref{eq:KKT_a}), we get
\begin{align}
\mathbf{Q}\mathbf{x}=\mathbf{q}.\label{eq:linear_sys}
\end{align}
Since $\mathbf{q}\perp\mathcal{N}(\mathbf{Q})$,  $\mathbf{q}$ belongs to the range space $\mathcal{R}(\mathbf{Q})$, which is the linear span of columns of $\mathbf{Q}$. 
By left multiplying $\mathbf{U}^H$ to both sides of (\ref{eq:linear_sys}), we obtain
\begin{align}
\mathbf{\Lambda}\mathbf{U}^H\mathbf{x}=\mathbf{p}.\label{eq:linear_sys_2}
\end{align}

Let $\bar{\mathbf{\Lambda}}$ be the top-left $r\times r$ sub-matrix of $\mathbf{\Lambda}$, $\bar{\mathbf{U}}$ be the left-most $r$ columns of $\mathbf{U}$, $\hat{\mathbf{U}}$ be the right-most $n\!-\!r$ columns of $\mathbf{U}$, and $\bar{\mathbf{p}}\triangleq[p_1,\cdots,p_r]^T$. We can then simplify (\ref{eq:linear_sys_2}) to
\begin{align}
\bar{\mathbf{\Lambda}}\bar{\mathbf{U}}^H\mathbf{x}=\bar{\mathbf{p}}.\label{eq:linear_sys_3}
\end{align}

Since $\mathbf{U}$ is unitary, we can represent $\mathbf{x}$ as a linear combination of its columns, i.e. $\mathbf{x}\triangleq\sum_{k=1}^{n}\alpha_k\mathbf{u}_k$. Recall the fact that $\mathcal{R}(\hat{\mathbf{U}})\!=\!\mathcal{N}(\mathcal{\mathbf{Q}})\perp\mathcal{R}(\mathcal{\mathbf{Q}})\!=\!\mathcal{R}(\bar{\mathbf{U}})$ and $\mathbf{q}\in\mathcal{R}(\mathbf{Q})$. This implies that the values $\{\alpha_k\}_{k=r\!+\!1}^n$ have no impact on the objective of the problem ($\mathsf{P}_{Cvx-TRSP}$). To obtain the solution with the minimal $l_2$-norm, $\{\alpha_k\}_{k=r\!+\!1}^n$ should therefore be set to zero. Now from (\ref{eq:linear_sys_3}), we  get 
\begin{align}
\mathbf{x}^{\star}=\bar{\mathbf{U}}\bar{\mathbf{\Lambda}}^{-1}\bar{\mathbf{p}}=\bar{\mathbf{U}}\bar{\mathbf{\Lambda}}^{-1}\bar{\mathbf{U}}^H\mathbf{q}=\mathbf{U}\mathbf{\Lambda}^{\dagger}\mathbf{U}^H\mathbf{q}=\mathbf{Q}^{\dagger}\mathbf{q},\label{eq:f_i_opt_mu_i0}
\end{align}

To verify that (\ref{eq:f_i_opt_mu_i0}) is indeed the solution, we need to test the KKT conditions. Obviously, $\mu_i=0$ naturally satisfies all the KKT conditions except (\ref{eq:KKT_b}). To verify (\ref{eq:KKT_b}), we utilize (\ref{eq:f_i_opt_mu_i0}) to obtain
\begin{align}
\|\mathbf{x}^{\star}\|_2^2=\bar{\mathbf{p}}^H\bar{\mathbf{\Lambda}}^{-1}\bar{\mathbf{U}}^H\bar{\mathbf{U}}\bar{\mathbf{\Lambda}}^{-1}\bar{\mathbf{p}}=\sum_{k=1}^{r}|p_{k}|^2\lambda_{k}^{-2}\leq\rho^2,
\label{eq:solution_case3}
\end{align}
where the above inequality follows from the fact that $\sum_{k=1}^r|p_i|^2\lambda_k^{-2}\leq\rho^2$ for $\mu=0$. Hence, all the KKT conditions are satisfied and the proof is thus complete.
\end{proof}

\subsection{Proof of Theorem \ref{thm:scalar_case}}
\label{subsec:appendix_thm_scalar_case}

\begin{proof}

We follow a similar outline of Theorem \ref{thm:sol_Cvx_TRSP} to solve ($\mathsf{P}4_i$). The key point leading to a closed-form solution is the fact that the quadratic matrix $\mathbf{H}_i^H\mathbf{g}\mathbf{g}^H\mathbf{H}_i$ has a rank of 1, i.e. $r_i=1$ in Theorem \ref{thm:sol_Cvx_TRSP}. Thus we obtain 
\begin{align}
&(\sigma_s^2+\sigma_i^2)\mathbf{H}_i^H\mathbf{g}\mathbf{g}^H\mathbf{H}_i\\
&=\mathbf{U}_i
\left[
\begin{array}{cc}
(\sigma_s^2+\sigma_i^2)\mathbf{g}^H\mathbf{H}_i\mathbf{H}_i^H\mathbf{g} & \mathbf{0}^H \\
\mathbf{0} & \mathbf{O}_{(KN_i-1)\times(KN_i-1)}
\end{array}
\right]
\mathbf{U}_i^H,\nonumber
\end{align}
with eigenvectors $\mathbf{U}_i\triangleq\big[\mathbf{u}_{i,1}, \mathbf{u}_{i,2},\cdots,\mathbf{u}_{i,N_i}\big]$ having its columns $\{\mathbf{u}_{i,j}\}_{j=1}^{N_i}$ satisfying the following properties
\begin{align}
\!\!\!\!\!\!\mathbf{u}_{i,1}\!=\!\frac{\mathbf{H}_i^H\mathbf{g}}{\|\mathbf{H}_i^H\mathbf{g}\|_2}, \ \mbox{and}\  \mathbf{u}_{i,j}^H\mathbf{H}_i^H\mathbf{g}\!=\! 0, \ \mbox{for\ } j\!=\!2,\cdots,N_i.
\end{align}
It can be readily checked that the projection of linear coefficient vector in the objective function along with $\mathbf{U}_i$, which we denote as $\mathbf{p}_i$ (and which is the counterpart of $\mathbf{p}$ in Theorem \ref{thm:sol_Cvx_TRSP}), is given by:
\begin{align}
\!\!\!\!\!\!\!\!\!\!\!\!p_{i,1}\!=\!\sigma_s^2(1-\mathbf{g}^H\widetilde{\mathbf{q}}_i)\|\mathbf{H}_i^H\mathbf{g}\|_2;\ p_{i,j}\!=\!0,j\!=\!\{2,\cdots,N_i\}.
\end{align}
At this time, the function $f(\mu_i)\triangleq\sum_{k=1}^{n}|p_k|^2(\lambda_k+\mu_i)^{-2}$  reduces to an elegant form
\begin{align}
f(\mu_i)=\frac{\sigma_s^4\big|1-\mathbf{g}^H\widetilde{\mathbf{q}}_i\big|^2\big\|\mathbf{H}_i^H\mathbf{g}\big\|^2_2}{\Big(\mu_i+(\sigma_s^2+\sigma_i^2)\mathbf{g}^H\mathbf{H}_i\mathbf{H}_i^H\mathbf{g}\Big)^2}.
\end{align}
Based on the above observations, the two cases for positive and zero $\mu_i^{\star}$ can be specified as follows:

If $\mu_i^{\star}>0$, then $\underset{\mu_i\rightarrow0}{\lim}f(\mu_i)>\bar{P}_i$, i.e. $\sigma_s^4\big|1-\mathbf{g}^H\widetilde{\mathbf{q}}_i\big|^2>(\sigma_s^2+\sigma_i^2)^2\bar{P}_i\|\mathbf{H}_i^H\mathbf{g}\|^2_2$. In this case, the optimal solution to ($\mathsf{P}4_i$) is given by
\begin{subequations}
\begin{align}
\!\!\!\!\!\!\!\!\mu_i^{\star}&\!=\!\sigma_s^2\bar{P}_i^{-\frac{1}{2}}\big|1\!-\!\mathbf{g}^H\widetilde{\mathbf{q}}_i\big|\big\|\mathbf{H}_i^H\mathbf{g}\big\|_2\!-\!(\sigma_s^2\!+\!\sigma_i^2)\big\|\mathbf{H}_i^H\mathbf{g}\big\|_2^2, \label{eq:opt_positive_mu_i_scalar_case}\\
\!\!\!\!\!\!\!\!\mathbf{f}_i^{\star}&\!=\!\sigma_s^2(1\!-\!\mathbf{g}^H\widetilde{\mathbf{q}}_i)\Big(\mu_i^{\star}\mathbf{I}\!+\!(\sigma_s^2\!+\!\sigma_i^2)\mathbf{H}_i^H\mathbf{g}\mathbf{g}^H\mathbf{H}_i\Big)^{-1}\!\!\mathbf{H}_i^H\mathbf{g}.\label{eq:opt_positive_f_i_scalar_case}
\end{align}
\end{subequations}


Otherwise, $\mu_i^{\star}=0$. This holds if and only if $\sigma_s^4\big|1\!-\!\mathbf{g}^H\widetilde{\mathbf{q}}_i\big|^2\!\leq\!(\sigma_s^2\!+\!\sigma_i^2)^2\bar{P}_i\|\mathbf{H}_i^H\mathbf{g}\|^2_2$ and the optimal $\mathbf{f}_i^{\star}$ is given by
\begin{align}
\!\!\!\!\mathbf{f}_i^{\star}\!=\!\frac{\sigma_s^2(1\!-\!\mathbf{g}^H\widetilde{\mathbf{q}}_i)\mathbf{H}_i^H\mathbf{g}}{(\sigma_s^2\!+\!\sigma_i^2)\mathbf{g}^H\mathbf{H}_i\mathbf{H}_i^H\mathbf{g}}.\label{eq:opt_zero_f_i_scalar_case}
\end{align}
Thus we have shown that, for scalar transmission, fully closed-form solution to ($\mathsf{P}4_i$) can be obtained without bisection search or eigenvalue decomposition.
\end{proof}

\subsection{Proof of Theorem \ref{thm:convergence_3BCA}}
\label{subsec:appendix_thm_convergnece_3BCA}

\begin{proof}
For each subproblem, we maximize the objective function with respect to a subset of variables with others being fixed, and hence the objective value (i.e. the MI sequence) should be non-decreasing as the iteration proceeds. 

Under the positive definiteness assumption of $\mathbf{\Sigma}_{\mathbf{s}}$, $\big(\mathbf{\Sigma}_{\mathbf{s}}\!\!+\!\!\mathbf{\Sigma}_i\big)\succ0$, we have,  $\forall i \in\{1,\cdots,L\} $, 
\begin{align}
\!\!\!\!\|\mathbf{F}_i\|^2_F\lambda_{\min}\big(\mathbf{\Sigma}_{\mathbf{s}}\!\!+\!\!\mathbf{\Sigma}_i\big)\leq\Tr\{\mathbf{F}_i\big(\mathbf{\Sigma}_{\mathbf{s}}\!\!+\!\!\mathbf{\Sigma}_i\big)\mathbf{F}_i^H\}\leq P_i,
\end{align}
where $\lambda_{\min}(\cdot)$ denotes the minimum eigenvalue of a Hermitian matrix. Since $\lambda_{\min}\big(\mathbf{\Sigma}_{\mathbf{s}}\!\!+\!\!\mathbf{\Sigma}_i\big)\!>\!0$, $\|\mathbf{F}_i\|^2_F$ is finite for all $i$, and the variable $\{\mathbf{F}_i\}_{i=1}^L$ is bounded. By Bolzano-Weierstrass Theorem, there exits a subsequence $\{k_j\}_{j=1}^{\infty}$ such that $\{\mathbf{F}_i^{(k_j)}\}_{i=1}^L$ converges. Since $\mathbf{G}$ and $\mathbf{W}$ are updated by continuous functions of $\{\mathbf{F}_i\}_{i=1}^L$ in (\ref{eq:optimal_G_MI}) and (\ref{eq:optimal_W_MI}), $\big(\{\mathbf{F}_i^{(k_j)}\}_{i=1}^L, \mathbf{W}^{(k_j)}, \mathbf{G}^{(k_j)}\big)$ converges. This proves the existence of limit points in the solution sequence.

Let $\big(\{\bar{\mathbf{F}}_i\}_{i=1}^L, \bar{\mathbf{W}}, \bar{\mathbf{G}}\big)$ be an arbitrary limit point of $\big(\{\mathbf{F}_i^{(k)}\}_{i=1}^L, \mathbf{W}^{(k)}, \mathbf{G}^{(k)}\big)$. There exists a subsequence $\{k_j\}$ such that $\big(\{\mathbf{F}_i^{(k_j)}\}_{i=1}^L, \mathbf{W}^{(k_j)}, \mathbf{G}^{(k_j)}\big)\stackrel{j\rightarrow\infty}{\longrightarrow} \big(\{\bar{\mathbf{F}}_i\}_{i=1}^L, \bar{\mathbf{W}}, \bar{\mathbf{G}}\big)$. Since $\{\mathbf{F}_i^{(k)}\}_{i=1}^L$ is bounded, by possibly restricting to a subsequence, we can assume that $\big(\{\mathbf{F}_i^{(k_j+1)}\}_{i=1}^L\big)$ converges to a limit $\big(\big\{\widehat{\mathbf{F}}_i\big\}_{i=1}^L\big)$. 

Note that for each $j$, $\{\mathbf{F}_i^{(k_j+1)}\}_{i=1}^L$ are feasible, i.e.
\begin{align}
\!\!\!\!\!\!\Tr\{\mathbf{F}_i^{(k_j+1)}\big(\mathbf{\Sigma}_{\mathbf{s}}\!\!+\!\!\mathbf{\Sigma}_i\big)\big(\mathbf{F}_i^{(k_j+1)}\big)^H)\}\leq P_i, i\in\{1,\cdots,L\}.
\end{align}
By taking $j\rightarrow\infty$ in the above inequalities, we obtain 
\begin{align}
\Tr\{\widehat{\mathbf{F}}_i\big(\mathbf{\Sigma}_{\mathbf{s}}\!\!+\!\!\mathbf{\Sigma}_i\big)\widehat{\mathbf{F}}_i^H)\}\leq P_i, \ \ i\in\{1,\cdots,L\}.
\end{align}
Hence $\big\{\widehat{\mathbf{F}}_i\big\}_{i=1}^L$ are feasible.

For any feasible $\big\{\mathbf{F}_i\big\}_{i=1}^L$, we have 
\begin{align}
\!\!\!\!\!\!\!\!\mathsf{MI}\big(\{\mathbf{F}_i\}_{i=1}^L\big|\mathbf{W}^{(\!k_j\!)}\!,\!\mathbf{G}^{(\!k_j\!)}\big)\!\leq\!\mathsf{MI}\big(\{\mathbf{F}_i^{(\!k_j\!+\!1\!)}\}_{i=1}^L\big|\mathbf{W}^{(\!k_j\!)}\!,\!\mathbf{G}^{(\!k_j\!)}\big).
\end{align}
Noticing that the MI function is continuous and taking $j\rightarrow\infty$ in the above equation, we obtain
\begin{align}
\mathsf{MI}\big(\{\mathbf{F}_i\}_{i=1}^L\big|\bar{\mathbf{W}},\bar{\mathbf{G}}\big)\leq\mathsf{MI}\big(\{\widehat{\mathbf{F}}_i\}_{i=1}^L\big|\bar{\mathbf{W}},\bar{\mathbf{G}}\big),
\end{align}
for any feasible $\{\mathbf{F}_i\}_{i=1}^L$. 

Now that the results $\{\mathbf{F}_i^{(k)}\}_{i=1}^L$ generated by Algorithm \ref{alg:3_block} are feasible, by continuity of power constraint functions, $\{\bar{\mathbf{F}}_i\}_{i=1}^L$ are feasible. Thus we have
\begin{align}\label{eq:ineq_to_eq_1}
\mathsf{MI}\big(\{\bar{\mathbf{F}}_i\}_{i=1}^L\big|\bar{\mathbf{W}},\bar{\mathbf{G}}\big)\leq\mathsf{MI}\big(\{\widehat{\mathbf{F}}_i\}_{i=1}^L\big|\bar{\mathbf{W}},\bar{\mathbf{G}}\big).
\end{align}
At the same time, since the MI sequence is non-decreasing and $\big(\{\bar{\mathbf{F}}_i\}_{i=1}^L, \bar{\mathbf{W}}, \bar{\mathbf{G}}\big)$ is a limit point of the solution sequence,  
\begin{align}\label{eq:ineq_to_eq_2}
\mathsf{MI}\big(\{\bar{\mathbf{F}}_i\}_{i=1}^L\big|\bar{\mathbf{W}},\bar{\mathbf{G}}\big)\geq\mathsf{MI}\big(\{\mathbf{F}_i^{(k)}\}_{i=1}^L\big|\bar{\mathbf{W}},\bar{\mathbf{G}}\big),
\end{align}
for any integer $k$. Substituting $k$ with $k_j$ in (\ref{eq:ineq_to_eq_2}), taking limit $j\rightarrow\infty$ and combining it with (\ref{eq:ineq_to_eq_1}), we can see that $\{\bar{\mathbf{F}}_i\}_{i=1}^L$ is actually an optimal solution to the problem ($\mathsf{P}2$) with parameters $\bar{\mathbf{W}}$ and $\bar{\mathbf{G}}$. Hence, $\{\bar{\mathbf{F}}_i\}_{i=1}^L$ satisfies the KKT conditions of ($\mathsf{P}2$) with parameters $\bar{\mathbf{W}}$ and $\bar{\mathbf{G}}$ given, as shown in (\ref{eq:KKT_P8}) on the top of next page for $i\in\{1,\cdots,L\}$.
\begin{figure*}[!t]
\normalsize
\setcounter{TempEqnCnt}{\value{equation}}
\setcounter{equation}{62}
\begin{subequations}
\label{eq:KKT_P8}
\begin{align}
\!\!\!\!\!\!\!\!-\mathbf{H}_i^H\bar{\mathbf{G}}\bar{\mathbf{W}}\Big(\mathbf{I}\!-\!\bar{\mathbf{G}}^H\big(\sum_{i=1}^L\mathbf{H}_i\bar{\mathbf{F}}_i\big)\Big)\mathbf{\Sigma}_{\mathbf{s}}\!+\!\mathbf{H}_i^H\bar{\mathbf{G}}\bar{\mathbf{W}}\bar{\mathbf{G}}^H\mathbf{H}_i\bar{\mathbf{F}}_i\mathbf{\Sigma}_i\!+\!\lambda_i\bar{\mathbf{F}}_i\big(\mathbf{\Sigma}_{\mathbf{s}}\!\!+\!\!\mathbf{\Sigma}_i\big)&=\mathbf{O}; \label{eq:KKT_P8_1}\\
\!\!\!\!\!\!\!\!\lambda_i\Big(\Tr\Big\{\bar{\mathbf{F}}_i\big(\mathbf{\Sigma}_{\mathbf{s}}\!\!+\!\!\mathbf{\Sigma}_i\big)\bar{\mathbf{F}}_i^H\Big\}-P_i\Big)&=0;\label{eq:KKT_P8_2}\\
\!\!\!\!\!\!\!\!\Tr\Big\{\bar{\mathbf{F}}_i\big(\mathbf{\Sigma}_{\mathbf{s}}\!\!+\!\!\mathbf{\Sigma}_i\big)\bar{\mathbf{F}}_i^H\Big\}&\leq P_i;\label{eq:KKT_P8_3}\\
\!\!\!\!\!\!\!\!\lambda_i&\geq 0.\label{eq:KKT_P8_4}
\end{align}
\end{subequations}
\setcounter{equation}{\value{TempEqnCnt}}
\hrulefill
\vspace*{4pt}
\end{figure*}
\setcounter{equation}{63}

To simplify the following exposition, we introduce two notations:
\begin{subequations}
\begin{align}
\bar{\mathbf{H}}&\triangleq\sum_{i=1}^L\mathbf{H}_i\bar{\mathbf{F}}_i;\\
\bar{\mathbf{\Sigma}}_{\mathbf{n}}&\triangleq\sigma_0^2\mathbf{I}+\sum_{i=1}^L\mathbf{H}_i\bar{\mathbf{F}}_i\mathbf{\Sigma}_i\bar{\mathbf{F}}_i^H\mathbf{H}_i^H. 
\end{align}
\end{subequations}

According to the update step in Algorithm \ref{alg:3_block}, the limit points $\bar{\mathbf{W}}$ and $\bar{\mathbf{G}}$ are related to $\big\{\bar{\mathbf{F}}_i\big\}_{i=1}^L$ in the following way:
\begin{subequations}
\label{eq:relation_GW_F}
\begin{align}
\bar{\mathbf{G}}&=\big[\bar{\mathbf{H}}\mathbf{\Sigma}_{\mathbf{s}}\bar{\mathbf{H}}^H\!\!+\!\!\bar{\mathbf{\Sigma}}_{\mathbf{n}}\big]^{-1}\bar{\mathbf{H}}\mathbf{\Sigma}_{\mathbf{s}}, \\
\bar{\mathbf{W}}&=\bar{\mathbf{H}}^H\bar{\mathbf{\Sigma}}_{\mathbf{n}}^{-1}\bar{\mathbf{H}}+\mathbf{\Sigma}_{\mathbf{s}}^{-1}.
\end{align} 
\end{subequations}

From (\ref{eq:relation_GW_F}) we can get two identities, shown in 
(\ref{eq:KKT_condition_1st_term}) and (\ref{eq:KKT_condition_2nd_term}) respectively in next page. 
\begin{figure*}[!t]
\normalsize
\setcounter{TempEqnCnt}{\value{equation}}
\setcounter{equation}{65}
\begin{align}
\!\!\!\!&\bar{\mathbf{G}}\bar{\mathbf{W}}\Big(\mathbf{I}\!\!-\!\!\bar{\mathbf{G}}^H\bar{\mathbf{H}}\Big)\!=\!\big(\bar{\mathbf{H}}\mathbf{\Sigma}_{\mathbf{s}}\bar{\mathbf{H}}^H\!\!+\!\!\bar{\mathbf{\Sigma}}_{\mathbf{n}}\big)^{-1}\bar{\mathbf{H}}\mathbf{\Sigma}_{\mathbf{s}}\big(\bar{\mathbf{H}}^H\bar{\mathbf{\Sigma}}_{\mathbf{n}}^{-1}\bar{\mathbf{H}}\!\!+\!\!\mathbf{\Sigma}_{\mathbf{s}}^{-1}\big)\Big(\mathbf{I}\!\!-\!\!\mathbf{\Sigma}_{\mathbf{s}}\bar{\mathbf{H}}^H\big(\bar{\mathbf{H}}\mathbf{\Sigma}_{\mathbf{s}}\bar{\mathbf{H}}^H\!\!+\!\!\bar{\mathbf{\Sigma}}_{\mathbf{n}}\big)^{-1}\bar{\mathbf{H}}\Big) \nonumber\\
\!\!\!\!&\!=\!\big(\bar{\mathbf{H}}\mathbf{\Sigma}_{\mathbf{s}}\bar{\mathbf{H}}^H\!\!+\!\!\bar{\mathbf{\Sigma}}_{\mathbf{n}}\big)^{-1}\bar{\mathbf{H}}\nonumber\\
\!\!\!\!&\qquad\!\!+\!\!\big(\bar{\mathbf{H}}\mathbf{\Sigma}_{\mathbf{s}}\bar{\mathbf{H}}^H\!\!+\!\!\bar{\mathbf{\Sigma}}_{\mathbf{n}}\big)^{-1}\bar{\mathbf{H}}\mathbf{\Sigma}_{\mathbf{s}}\Big[\underbrace{\bar{\mathbf{H}}^H\bar{\mathbf{\Sigma}}_{\mathbf{n}}^{-1}\big(\bar{\mathbf{H}}\mathbf{\Sigma}_{\mathbf{s}}\bar{\mathbf{H}}^H\!\!+\!\!\bar{\mathbf{\Sigma}}_{\mathbf{n}}\big)\!\!-\!\!\big(\bar{\mathbf{H}}^H\bar{\mathbf{\Sigma}}_{\mathbf{n}}^{-1}\bar{\mathbf{H}}\!\!+\!\!\mathbf{\Sigma}_{\mathbf{s}}^{-1}\big)\mathbf{\Sigma}_{\mathbf{s}}\bar{\mathbf{H}}^H}_{=\mathbf{O}}\Big]\big(\bar{\mathbf{H}}\mathbf{\Sigma}_{\mathbf{s}}\bar{\mathbf{H}}^H\!\!+\!\!\bar{\mathbf{\Sigma}}_{\mathbf{n}}\big)^{-1}\bar{\mathbf{H}}\nonumber\\
\!\!\!\!&\!=\!\big(\bar{\mathbf{H}}\mathbf{\Sigma}_{\mathbf{s}}\bar{\mathbf{H}}^H\!\!+\!\!\bar{\mathbf{\Sigma}}_{\mathbf{n}}\big)^{-1}\bar{\mathbf{H}}\label{eq:KKT_condition_1st_term}
\end{align} 
\setcounter{equation}{\value{TempEqnCnt}}
\hrulefill
\vspace*{4pt}
\end{figure*}
\setcounter{equation}{66}

\begin{figure*}[!t]
\normalsize
\setcounter{TempEqnCnt}{\value{equation}}
\setcounter{equation}{66}
\begin{align}
\bar{\mathbf{G}}\bar{\mathbf{W}}\bar{\mathbf{G}}^H&\!=\!\big(\bar{\mathbf{H}}\mathbf{\Sigma}_{\mathbf{s}}\bar{\mathbf{H}}^H\!\!+\!\!\bar{\mathbf{\Sigma}}_{\mathbf{n}}\big)^{-1}\bar{\mathbf{H}}\mathbf{\Sigma}_{\mathbf{s}}\big(\bar{\mathbf{H}}^H\bar{\mathbf{\Sigma}}_{\mathbf{n}}^{-1}\bar{\mathbf{H}}\!\!+\!\!\mathbf{\Sigma}_{\mathbf{s}}^{-1}\big)\mathbf{\Sigma}_{\mathbf{s}}\bar{\mathbf{H}}^H\big(\bar{\mathbf{H}}\mathbf{\Sigma}_{\mathbf{s}}\bar{\mathbf{H}}^H\!\!+\!\!\bar{\mathbf{\Sigma}}_{\mathbf{n}}\big)^{-1}\nonumber\\
&\!=\!\big(\bar{\mathbf{H}}\mathbf{\Sigma}_{\mathbf{s}}\bar{\mathbf{H}}^H\!\!+\!\!\bar{\mathbf{\Sigma}}_{\mathbf{n}}\big)^{-1}\Big[\bar{\mathbf{H}}\mathbf{\Sigma}_{\mathbf{s}}\bar{\mathbf{H}}^H\bar{\mathbf{\Sigma}}_{\mathbf{n}}^{-1}\bar{\mathbf{H}}\mathbf{\Sigma}_{\mathbf{s}}\bar{\mathbf{H}}^H\!\!+\!\!\bar{\mathbf{H}}\mathbf{\Sigma}_{\mathbf{s}}\bar{\mathbf{H}}^H\Big]\big(\bar{\mathbf{H}}\mathbf{\Sigma}_{\mathbf{s}}\bar{\mathbf{H}}^H\!\!+\!\!\bar{\mathbf{\Sigma}}_{\mathbf{n}}\big)^{-1}\label{eq:KKT_condition_2nd_term}\nonumber\\
&\!=\!\big(\bar{\mathbf{H}}\mathbf{\Sigma}_{\mathbf{s}}\bar{\mathbf{H}}^H\!\!+\!\!\bar{\mathbf{\Sigma}}_{\mathbf{n}}\big)^{-1}\Big[\bar{\mathbf{H}}\mathbf{\Sigma}_{\mathbf{s}}\bar{\mathbf{H}}^H\bar{\mathbf{\Sigma}}_{\mathbf{n}}^{-1}\big(\bar{\mathbf{H}}\mathbf{\Sigma}_{\mathbf{s}}\bar{\mathbf{H}}^H\!\!+\!\!\bar{\mathbf{\Sigma}}_{\mathbf{n}}\big)\Big]\big(\bar{\mathbf{H}}\mathbf{\Sigma}_{\mathbf{s}}\bar{\mathbf{H}}^H\!\!+\!\!\bar{\mathbf{\Sigma}}_{\mathbf{n}}\big)^{-1}\nonumber \\
&\!=\!\big(\bar{\mathbf{H}}\mathbf{\Sigma}_{\mathbf{s}}\bar{\mathbf{H}}^H\!\!+\!\!\bar{\mathbf{\Sigma}}_{\mathbf{n}}\big)^{-1}\bar{\mathbf{H}}\mathbf{\Sigma}_{\mathbf{s}}\bar{\mathbf{H}}^H\bar{\mathbf{\Sigma}}_{\mathbf{n}}^{-1}
\end{align} 
\setcounter{equation}{\value{TempEqnCnt}}
\hrulefill
\vspace*{4pt}
\end{figure*}
\setcounter{equation}{67}

Substituting (\ref{eq:KKT_condition_1st_term}) and (\ref{eq:KKT_condition_2nd_term}) into (\ref{eq:KKT_P8_1}), we can rewrite the first-order KKT conditions associated with only $\{\bar{\mathbf{F}}_i\}_{i=1}^L$  as  shown in (\ref{eq:KKT_P0_MI}) in next page.
\begin{figure*}[!t]
\normalsize
\setcounter{TempEqnCnt}{\value{equation}}
\setcounter{equation}{67}
\begin{align}
\label{eq:KKT_P0_MI}
\!\!\!\!\!\!\!\!\!\!\!\!\mathbf{H}_i^H\Big[\big(\sigma_0^2\mathbf{I}\!\!+\!\!\sum_{i=1}^L\mathbf{H}_i\bar{\mathbf{F}}_i\mathbf{\Sigma}_i\bar{\mathbf{F}}_i^H\mathbf{H}_i^H\big)\!\!+\!\!\big(\sum_{i=1}^L\mathbf{H}_i\bar{\mathbf{F}}_i\big)
\mathbf{\Sigma}_{\mathbf{s}}\big(\sum_{i=1}^L\mathbf{H}_i\bar{\mathbf{F}}_i\big)^H
\Big]^{-1}\big(\sum_{i=1}^L\mathbf{H}_i\bar{\mathbf{F}}_i\big)\mathbf{\Sigma}_{\mathbf{s}}\Big[\mathbf{I}&\\
\!\!\!\!\!\!\!\!\!\!\!\!-\big(\sum_{i=1}^L\mathbf{H}_i\bar{\mathbf{F}}_i\big)^H\big(\sigma_0^2\mathbf{I}+\sum_{i=1}^L\mathbf{H}_i\bar{\mathbf{F}}_i\mathbf{\Sigma}_i\bar{\mathbf{F}}_i^H\mathbf{H}_i^H\big)^{-1}\mathbf{H}_i\bar{\mathbf{F}}_i\mathbf{\Sigma}_i\Big]-\lambda_i\bar{\mathbf{F}}_i\big(\mathbf{\Sigma}_{\mathbf{s}}\!\!+\!\!\mathbf{\Sigma}_i\big)&=\mathbf{O},\ i\in\{1,\cdots,L\};\nonumber
\end{align}
\setcounter{equation}{\value{TempEqnCnt}}
\hrulefill
\vspace*{4pt}
\end{figure*}
\setcounter{equation}{68}

To check the conditions of the original problem ($\mathsf{P}0$), we need to determine the derivative of its Lagrangian function, or, equivalently, the derivative of $\mathsf{MI}$ with respect to $\{\mathbf{F}_i\}$. By defining 
\begin{align}
\mathbf{H}\triangleq\sum_{i=1}^L\mathbf{H}_i\mathbf{F}_i, 
\end{align} 
the derivative of $\mathsf{MI}$ can be calculated and is shown in (\ref{eq:derivative_MI}) in next page, with $C_1(\mathrm{d}\mathbf{F}_i)$ and $C_2(\mathrm{d}\mathbf{F}_i)$ being uninteresting terms involving only $\mathrm{d}\mathbf{F}_i$ and independent of $\mathrm{d}(\mathbf{F}_i^*)$.
\begin{figure*}[!t]
\normalsize
\setcounter{TempEqnCnt}{\value{equation}}
\setcounter{equation}{69}
\begin{subequations}
\label{eq:derivative_MI}
\begin{align}
\label{eq:derivative_MI_1}
\mathrm{d}(\mathsf{MI})&=\Tr\Big\{\big(\mathbf{I}\!+\!\mathbf{H}\mathbf{\Sigma}_{\mathbf{s}}\mathbf{H}^H\mathbf{\Sigma}_{\mathbf{n}}^{-1}\big)^{-1}\mathrm{d}\big(\mathbf{H}\mathbf{\Sigma}_{\mathbf{s}}\mathbf{H}^H\mathbf{\Sigma}_{\mathbf{n}}^{-1}\big)\Big\} \nonumber\\
&=\Tr\Big\{\big(\mathbf{I}\!+\!\mathbf{H}\mathbf{\Sigma}_{\mathbf{s}}\mathbf{H}^H\mathbf{\Sigma}_{\mathbf{n}}^{-1}\big)^{-1}\Big[\mathbf{H}\mathbf{\Sigma}_{\mathbf{s}}\mathrm{d}(\mathbf{H}^H)\mathbf{\Sigma}_{\mathbf{n}}^{-1}\!+\!\mathbf{H}\mathbf{\Sigma}_{\mathbf{s}}\mathbf{H}^H\mathrm{d}(\mathbf{\Sigma}_{\mathbf{n}}^{-1})\Big]\Big\}+C_1(\mathrm{d}\mathbf{F}_i)\nonumber \\
&=\Tr\Big\{\mathbf{H}_i^H\big(\mathbf{\Sigma}_{\mathbf{n}}\!\!+\!\!\mathbf{H}\mathbf{\Sigma}_{\mathbf{s}}\mathbf{H}^H\big)^{-1}\mathbf{H}\mathbf{\Sigma}_{\mathbf{s}}\Big[\mathbf{I}-\mathbf{H}^H\mathbf{\Sigma}_{\mathbf{n}}^{-1}\mathbf{H}_i\mathbf{F}_i\mathbf{\Sigma}_i\Big]\mathrm{d}(\mathbf{F}_i)^H\Big\}+C_2(\mathrm{d}\mathbf{F}_i),
\end{align}
\begin{align}
\label{eq:derivative_MI_2}
\Rightarrow \frac{\partial\mathsf{MI}}{\partial\mathbf{F}_i^{\ast}}\!&=\!\mathbf{H}_i^H\Big[\big(\sigma_0^2\mathbf{I}\!\!+\!\!\sum_{i=1}^L\mathbf{H}_i\mathbf{F}_i\mathbf{\Sigma}_i\mathbf{F}_i^H\mathbf{H}_i^H\big)\!\!+\!\!\big(\!\sum_{i=1}^L\mathbf{H}_i\mathbf{F}_i\!\big)
\mathbf{\Sigma}_{\mathbf{s}}\big(\!\sum_{i=1}^L\mathbf{H}_i\mathbf{F}_i\!\big)^H
\Big]^{-1}\!\!\!\big(\!\sum_{i=1}^L\mathbf{H}_i\mathbf{F}_i\!\big)\mathbf{\Sigma}_{\mathbf{s}}\Big[\mathbf{I}-\nonumber\\
&\qquad\qquad\qquad\big(\!\sum_{i=1}^L\mathbf{H}_i\mathbf{F}_i\!\big)^H\big(\sigma_0^2\mathbf{I}\!\!+\!\!\sum_{i=1}^L\mathbf{H}_i\mathbf{F}_i\mathbf{\Sigma}_i\mathbf{F}_i^H\mathbf{H}_i^H\big)^{-1}\mathbf{H}_i\mathbf{F}_i\mathbf{\Sigma}_i\Big], \ i\in\{1,\cdots,L\}.
\end{align}
\end{subequations}
\setcounter{equation}{\value{TempEqnCnt}}
\hrulefill
\vspace*{4pt}
\end{figure*}
\setcounter{equation}{70}

Comparing (\ref{eq:KKT_P0_MI}) with the derivative in (\ref{eq:derivative_MI_2}), it is easy to see that (\ref{eq:KKT_P0_MI}) is actually the first-order KKT condition of ($\mathsf{P}0$) that optimizes $\mathsf{MI}$. Together with (\ref{eq:KKT_P8_2}), (\ref{eq:KKT_P8_3}) and (\ref{eq:KKT_P8_4}), we see that the KKT conditions of the original problem are satisfied by $\{\bar{\mathbf{F}}_i\}_{i=1}^{L}$. This completes the entire proof.
\end{proof}

\end{document}